\documentclass[11pt,authoryear]{elsarticle}

\makeatletter
\def\ps@pprintTitle{%
 \let\@oddhead\@empty
 \let\@evenhead\@empty
 \def\@oddfoot{}%
 \let\@evenfoot\@oddfoot}
\makeatother

\usepackage[margin=0.9in]{geometry}
\usepackage{setspace}
\usepackage[toc]{appendix}
\RequirePackage[OT1]{fontenc}

\usepackage{amsmath,amssymb,amsthm,bm,latexsym}
\usepackage{graphicx,psfrag,epsf}
\usepackage{epigraph,xcolor}
\usepackage{enumerate}
\usepackage{subcaption}
\usepackage{amsfonts}
\usepackage{float}
\usepackage{orcidlink}
\usepackage{eurosym}
\usepackage{algpseudocode}
\usepackage{algorithm}
\usepackage{arydshln}
\usepackage{easyReview} 
\usepackage{tikz}
\usepackage{enumitem}
\usepackage{adjustbox}
\usepackage{rotating}
\usepackage{url}
\usepackage{booktabs}
\usepackage{multirow}
\usepackage{hyperref}
\usepackage[capitalise]{cleveref}

\hypersetup{
    colorlinks=true,
    linkcolor=purple,
    citecolor=blue,
    urlcolor=purple,
    filecolor=black,
    menucolor=black,
    anchorcolor=black
  }

\allowdisplaybreaks

\def\I{\mathbf{I}}

\def\X{\mathbb{X}}

\def\F{\mathcal{F}}

\def\D{\mathcal{D}}
\def\var{{\textrm{Var}}\,}

\def\N{\mathbb{N}}
\def\Z{\mathbb{Z}}

\def\R{\mathbb{R}}

\def\W{\mathbb{W}}
\def\w{\bm{w}}
\def\cP{\mathcal{P}}
\def\cC{\mathcal{C}}

\newcommand{\gauss}{\mathcal{N}}

\newcommand{\ind}{\mathbb{I}}

\newcommand{\x}{\bm{X}}
\newcommand{\RNum}[1]{\uppercase\expandafter{\romannumeral #1\relax}}

\newcommand{\abs}[1]{\left\lvert #1 \right\rvert}

\newcommand{\OLS}{\textsf{OLS}}
\newcommand{\proposed}{\textsf{PROPOSED}}
\newcommand{\enet}{\textsf{ENET}}
\newcommand{\lasso}{\textsf{LASSO}}
\newcommand{\tobit}{\textsf{TOBIT}}
\newcommand{\xgb}{\textsf{XGBOOST}}

\newtheoremstyle{general}
{3mm} 
{3mm} 
{\it} 
{} 
{\bfseries} 
{.} 
{.5em} 
{} 

\theoremstyle{general}

\newtheorem{proposition}{Proposition}

\newtheorem{remark}{Remark}

\begin{document}
  
\begin{frontmatter}
\title{Optimising football transfer strategy under budget constraints: A weighted multi-criteria approach}

\author[NU]{Tathagata Basu\orcidlink{0000-0002-6851-154X}}\ead{tathagata.basu@newcastle.ac.uk}

\author[IIMB]{Soudeep Deb\orcidlink{0000-0003-0567-7339}}\ead{soudeep@iimb.ac.in}

\author[UoE]{Rishideep Roy\orcidlink{0000-0003-1883-7848}}\ead{rishideep.roy@essex.ac.uk}

\address[NU]{Newcastle University, Newcastle upon Tyne, NE1 7RU, United Kingdom.}

\address[IIMB]{Indian Institute of Management Bangalore, Bannerghatta Road, Bengaluru, KA 560076, India.}

\address[UoE]{University of Essex, Wivenhoe Park, Colchester, Essex CO4 3SQ, United Kingdom.}


\begin{abstract}
The football transfer market is a complex, dynamic environment in which clubs compete to acquire players who strengthen their squads. While several frameworks estimate a player’s worth, a comprehensive approach that captures both squad optimisation and transfer market dynamics remains limited. In this paper, we propose a quantitative framework for optimising football transfer strategy under budget constraints, integrated with a competitive bidding paradigm. Using data from professional football leagues, we construct player performance and transfer price models using linear mixed-effects frameworks that incorporate player characteristics, recent performance, team context, and league effects. The predicted ratings and estimated transfer prices are then integrated into a weighted multi-criteria constrained optimisation framework that determines a club’s transfer activities at the end of the season. Finally, these optimal transfer decisions are embedded within an independent private-value auction model with a random reserve price to analyse market behaviour when multiple teams compete for the same player. We illustrate our approach using the 2018–19 season of the English Premier League to demonstrate its ability to capture transfer-market dynamics.
\end{abstract}

\end{frontmatter}

\newpage

\section{Introduction}\label{sec:introduction}

Association football (though known as soccer in Australia, Canada, United States of America etc., for convenience we shall use `football' throughout this paper) has been one of the most celebrated sports globally, generating viewership from all over the world. Among different football leagues, the English Premier League (EPL) is arguably the most popular in the world. At least $8$ of the EPL clubs are worth more than $\text{\euro} 600$ million\footnote{Source: \url{https://en.wikipedia.org/wiki/Forbes_list_of_the_most_valuable_football_clubs}}. Consequently, there is a lot of interest and financial involvement in the selection of squads of these clubs. The squads for these teams are dynamic, and players are bought and sold frequently. For example, in the ongoing 2025-26 season, a total of $167$ players have been bought, and $194$ players have been sold by the EPL clubs, amounting to a total of nearly $\text{\euro} 6.5$ billion\footnote{Source: \url{https://www.transfermarkt.co.uk/}}. \cite{poli2022demographic} showed that the average tenure at a club ranges between $(1.87, \, 2.56)$ years for professional players across Europe, implying that a transaction happens for every player once in less than three years on average. Naturally, an optimum strategy for transfer is paramount in football, and we propose a novel multi-criteria optimisation method under appropriate constraints in this regard. 

The proposed methodology is described in \Cref{sec:methodology}. Before that, we find it imperative to present a short discussion (\Cref{sec:auctions-drafts-transfers}) of auctions, drafts, and transfers, which are three different mechanisms of buying/selling players in sports. It is followed by a relevant review of extant literature (\Cref{sec:lit-review}) and how we contribute to it (\Cref{sec:contribution}). The data and detailed results are illustrated in \Cref{sec:results}. We conclude the paper with important remarks and future directions in \Cref{sec:conclusions}.

\subsection{Auctions, drafts, and transfers}\label{sec:auctions-drafts-transfers}

In the realm of sports, auctions, drafts, and transfers represent three distinct mechanisms for team building and resource allocation, each with unique operational characteristics. Auctions are structured, competitive events where teams bid for players within a fixed budget. It has become quite common in cricket, in competitions like the Indian Premier League in India or the Big Bash League in Australia. These mechanisms operate as combinatorial optimisation problems under time constraints, requiring rapid decision-making and strategic adjustments as bids evolve. \cite{singh2011dynamic, ahmed2013multi} discussed a couple of different auction strategies in this regard; whereas \cite{chandrakar2021measuring, chittibabu2023base} developed various statistical and machine learning techniques to determine appropriate values for the players in the auction. In contrast, drafts are widely used in North American sports, such as American Football and basketball. Draft systems assign player selection rights based on predefined rules, such as the reverse order of league standings, to promote competitive balance. Unlike auctions, drafts are less about financial bidding and more about strategic planning. They can be viewed as sequential decision problems, where the outcome of one team's choice influences subsequent decisions, creating a multi-stage optimisation framework. The reader is referred to \cite{swartz2011drafts} for discussions on drafts and their comparison with auctions. 

Our focus in this paper, however, is on the aspect of ``transfers'' which is common in football and operates in a more decentralised fashion. Transfers involve direct negotiations between clubs about the fee, while salaries and a few other specific clauses are typically negotiated individually. The transfer window imposes a time-bound constraint, but unlike auctions or drafts, the process is not inherently competitive among all teams for the same player. Transfers represent a multi-objective optimisation problem where clubs must balance financial constraints, tactical needs, and player availability. Decision-making here is deeply influenced by market dynamics, scouting analytics, and long-term team strategy, making it a highly flexible yet complex resource allocation problem. We review some of the existing works in this domain in the next subsection. Interestingly, it must be noted that football is the only sport where the aspect of transfer is paramount and holds the key to success. Among other sports, transfers have been observed in rugby \citep{groeneveld2007matters}, though there is no formal transfer market yet. In the recent past, the NCAA or the National Collegiate Athletic Association also introduced the concepts of transfer portals, but it met with strong criticisms due to various reasons \citep{reese2023unintended}.

\subsection{Review of literature on transfer strategies}\label{sec:lit-review}

There are broadly two classes of problems which are of relevance: identification of the right players to be transferred, and appropriate valuation of them to make financially judicious decisions. From the perspective of our main objective, both problems are important, and it is of the essence here to discuss extant literature on these two aspects.

The empirical literature on football player valuation and transfer pricing has expanded rapidly over the past decade, reflecting both the growth of the global transfer market and the increasing availability of detailed performance and financial data. Early contributions modelled player values or transfer fees using player characteristics, age, position, nationality, experience, recent performance, and club-level factors. For example, \cite{tunaru2005option} treated footballers as real options and derived a theoretically grounded financial value using performance indices and club turnover, while \cite{carmichael1999labour} emphasised the role of bargaining, player productivity, club status, and market power in shaping transfer outcomes. Subsequent empirical studies \citep[e.g.,][]{ruijg2015determinants, Mueller2017, Metelski2021} further showed that player attributes, positional roles, recent performance, team characteristics, and league context are important determinants of transfer fees. In parallel, multiple approaches have also been proposed to model player ratings and evaluate their performance. \citet{BAKER2018659} constructed time-varying team ratings for international football using cross-validated shrinkage methods while \citet{KHARRAT2020726} developed plus-minus player ratings for football using regularised regression, both demonstrating that structured statistical rating systems can capture meaningful variation in quality beyond simple performance counts.

In terms of improving predictive accuracy for transfer fees or market values, focus has been given on incorporating richer performance metrics and more flexible statistical or machine learning models. \cite{campa2022exploring} studied systematic differences between transfer fees and estimated inherent values, highlighting the importance of factors such as transfer-window timing, league match, playing position, contract expiry, and age. \citet{McHaleHolmes2023} combined advanced event-based metrics with machine-learning methods to estimate transfer fees, while \citet{Poli2024} proposed a global statistical model that combines sporting, biographical, and contractual variables. \citet{COATES2022523} evaluated crowd-sourced valuations as proxies for unobserved market prices and found systematic biases relative to realised transfer fees, highlighting the importance of modelling actual transaction prices directly. Complementary work by \citet{Thrane2024} showed that composite performance indices can provide additional explanatory power for market values. More broadly, \citet{FranceschiEtAl2024} provide a systematic review of player valuation research, while \citet{HillEtAl2025} placed football player valuation within wider approaches such as intrinsic, relative, real-options, and probabilistic valuation methods. Together, these emphasise that the extant literature primarily treats valuation as a pricing problem, rather than a decision problem embedded in a club's transfer strategy.

On that note, several recent studies have recast transfers as investment decisions under uncertainty. \citet{FollertGleissner2024}, for example, developed a decision-oriented valuation model in which a transfer is treated as an investment project generating uncertain sporting and financial benefits. Related work by \citet{KhalifeEtAl2025, Liu2025} highlighted how career stage, expected future performance, role-specific metrics, and post-prime depreciation patterns influence player valuation. These studies move beyond purely descriptive pricing models and emphasise the decision-making context in which transfer valuations are used. A closely related stream of research studies football squad construction directly as an optimisation problem. \cite{pantuso2017football}, for instance, formulated the football team composition problem as a stochastic programming model, where the objective is to maximise the expected market value of the team while accounting for budget limits, competition regulations, and uncertainty in future player values. Extending this line, \cite{pantuso2021maximizing} developed a chance-constrained model for football transfer-market decisions that seeks to construct a high-performing team while adapting to different budgets and financial-risk profiles. More recently, \cite{dantas2025squad} proposed an integer programming model for squad optimisation that incorporates tactical requirements and evaluates the approach using EPL clubs. 

Despite this progress, the link between valuation models, squad quality and transfer strategy remains relatively weak. The above studies demonstrate the relevance of optimisation-based decision support for football squad construction; however they either take the club’s strategic context as given or treat players independently, without considering positional complementarities, financial constraints or risk diversification. The present paper contributes to this strand of literature by shifting the focus from isolated valuations to transfer strategy as a dynamic portfolio problem. Building on insights of the valuation and investment-based strands, we develop a framework that explicitly links a club’s financial constraints and sporting objectives to its pattern of acquisitions and disposals over time. In doing so, we position transfer decisions within a coherent strategic model that can be used to evaluate alternative policies (e.g.\ focusing on development and resale, buying established prime-age players, or targeting undervalued segments), rather than only benchmarking individual transfer prices ex post.

\subsection{Our contribution}\label{sec:contribution}

It must be evident from the above discussions that transfers in football present a multi-criteria optimisation problem. Clubs aim to maximise squad strength while adhering to financial constraints, such as budget caps or UEFA’s Financial Fair Play regulations\footnote{Link: \url{https://en.wikipedia.org/wiki/UEFA_Financial_Fair_Play_Regulations}}. Uncertainty is inherent in transfers, with factors like fluctuating market values, player injuries, and competition from rival clubs affecting decision outcomes. Advanced techniques, including stochastic optimisation and game theory, can be applied to model these uncertainties and develop robust transfer strategies. By studying transfers through the lens of operational research, clubs can better navigate the complexities of the market and make decisions that optimise both short-term performance and long-term sustainability. We focus on the transfer-fee component of squad reallocation rather than the full accounting cost of a player acquisition. In practice, clubs also account for wages, contract length, signing bonuses, agent fees, amortisation, release clauses, and home-grown registration rules. These variables are central to club-level decision-making, but they are not consistently observable across leagues and seasons at the scale required for our empirical application. Moreover, in most cases, a player's transfer fee is indicative of their wage, and therefore, the wage negotiation and other incurred costs are only discussed after the seller club agrees to sell a player. We therefore treat transfer fees as the principal observable market-clearing cost and develop a modular optimisation framework into which additional financial terms can be incorporated when reliable data becomes available.

In this article, we first handle the questions of expected performance and transfer valuation of various players, should they join a focal club. While these have been explored before, we present comprehensive models that take into account appropriate variables related to the players and have been proven to outperform various other methods from existing studies. Next, we focus on developing an optimal transfer strategy from the perspective of a club, involving the buying and selling of players, to create an efficient squad under suitable constraints. This naturally leads to situations where particular players are targeted by multiple clubs, which in turn leads us to investigate an auction mechanism to summarise transfer market dynamics. We find it imperative to highlight that, to the best of our knowledge, such a comprehensive analysis of transfer strategy has not been done before.

\section{Proposed methodology}\label{sec:methodology}

Throughout this article, we shall use $\Z,\N,\R$ to denote the set of integers, the set of natural numbers, and the set of real numbers, respectively. For any $n \in \N$, we shall denote the set $\{1, \ldots, n \}$ by $[n]$. The notations $\I$ and $\bm{0}$ are used to indicate an identity matrix and a vector of all zeros (of appropriate order). A Gaussian distribution with mean $\bm{\gamma}$ and dispersion $\Gamma$ is denoted by $\gauss(\bm{\gamma},\Gamma)$.

To describe the main methodology, consider a collection of teams $\mathcal{C}$. With the objective of building a strategy for a team $c_0 \in \mathcal{C}$, let $\cP_{c_0}$ denote the index set of players who are already in $c_0$. We consider $\cP = \cup_{c\in\mathcal{C}} \, \cP_c$ as the complete index set of $N$ players ($N \in \N$) who can be considered for $c_0$ (including players already in $\cP_{c_0}$). Then, the objective is to design a strategy which deduces the optimum set of binary decision variables $\{x_i, \; i \in [N]\}$, i.e., whether the $i^{th}$ player should be selected in the team $c_0$, under relevant constraints. Without loss of generality, we assume that at the start, for all $i \in [N]$, there exists $\Tilde{c} \in \cC$ such that the $i^{th}$ player is in the squad for team $\Tilde{c}$. Let $Y_{i, c \to c', s}$ be the selling price of the $i^{th}$ player if they are sold from team $c$ to team $c'$ after the completion of season $s$, or during season $s+1$. This is a random variable, and we adopt a linear mixed-effects modelling strategy for this. Also, associated with the $i^{th}$ player, we let $R_{ics}$ denote the importance or value of the player in team $c$ during season $s$. This term is a metric that can quantify the contribution of the player to the team, and we shall rely on a rating system to model this. The details are elaborated later in the article. Hereafter, we use the terms value and rating interchangeably to indicate the same variable $R$.

From the viewpoint of practical implementation, our optimisation procedure works in three steps. First, we utilise a suitable model to forecast the value of every player in the entire pool if they are part of the team $c_0$ in the next season. This is explicated in \Cref{sec:model-rating}. Second, we use another appropriate regression framework (see \Cref{sec:model-transfer}) to predict the transfer fee that the team $c_0$ needs to pay for players outside their current squad, i.e., in $\cP \setminus \cP_{c_0}$. The same model is also utilised to forecast the minimum expected selling price for any player from $\cP$. Subsequently, in the next step of the procedure, we develop an appropriate optimisation framework to suggest which players to buy and which players to sell, while adhering to the required constraints. This part of the algorithm is explained in \Cref{sec:optimisation}. Finally, to get an insight on the transfer market dynamics, in \Cref{sec:bidding}, we formulate an independent private value auction framework with random reserve price to understand the effect on the final transfer price of a player when multiple clubs are recommended to buy them.

\subsection{Modelling the value/rating of players}\label{sec:model-rating}

To describe our modelling procedure, we use $i$ to index the players, $c$ for the teams, $\ell$ for the leagues, and $s$ for the seasons. For a player $i$ rostered at team $c$ during season $s$, let $R_{ics}\in\R$ denote the rating. As mentioned earlier, the first step of our pipeline constructs the counterfactual forecast for the focal club $c_0$ (denote it as $\hat{R}_{ic_0,s+1}$) by modelling the rating values in a particular season as a function of all information available at the end of the previous season. These forecasts are later going to be used as player-value inputs for the optimisation algorithm in \Cref{sec:optimisation}.

The model we propose is in the form of a linear mixed-effects model that captures (i) systematic effects of age, anthropometrics, recent form, team, and position; and (ii) latent heterogeneity associated with origin-destination club corridors and league environments. Specifically, the model structure is
\begin{equation*}
    R_{ics} = \mu_{ics} + u_{(c_{s-1}\to c_s)} + u^{\mathrm{cur}}_{\ell(c_s)} + u^{\mathrm{last}}_{\ell(c_{s-1})} + \varepsilon_{ics},
\end{equation*}
where $\mu_{ics}$ is the average rating that is expressed in an additive form of the effects of different features, $c_s$ denotes the club the player is part of in season $s$, the term $u_{(c_{s-1}\to c_s)}$ is a random intercept for the origin-destination club corridor corresponding to the specific observation, $u^{\mathrm{cur}}_{\ell(c_s)}$, $u^{\mathrm{last}}_{\ell(c_{s-1})}$ are random intercepts, respectively for the leagues of the current club and the last club the player has played in; and $\varepsilon_{ics}$ is white noise. Note that $c_{s-1}=c_s$ for players who have not been transferred, and that is handled appropriately in our model. We assume mutual independence for the random effects and let
\begin{equation*}
    u_{(c_{s-1}\to c_s)}\sim\mathcal{N}(0,\sigma^2_{\mathrm{club}}),\quad
    u^{\mathrm{cur}}_{\ell}\sim\mathcal{N}(0,\sigma^2_{\mathrm{cur}}),\quad
    u^{\mathrm{last}}_{\ell}\sim\mathcal{N}(0,\sigma^2_{\mathrm{last}}),\quad
    \varepsilon_{ics}\sim\mathcal{N}(0,\sigma^2).
\end{equation*}
The average rating is modelled as the following linear term
\begin{equation}\label{eq:mean-structure-rating}
\begin{split}
    \mu_{ics} &= \beta_0 + \left(\beta_{a1}A_{is} + \beta_{a2}A_{is}^2\right) + \beta_h H_{is} + \beta_w W_{is} + \bm{\beta}_{\mathrm{pos}}^{\top}\ind\{P_i=\cdot\} + \beta_{\mathrm{lr}}\,R_{i c_{s-1},\, s-1} + \beta_{\mathrm{tr}}\,\mathrm{TR}_{c,s} \\
    &\quad + \beta_{\mathrm{tr,sp}}\,\mathrm{TR}^{(\mathrm{pos})}_{c,s} + \beta_{\mathrm{td,sp}}\,\mathrm{TD}^{(\mathrm{pos})}_{c,s} + \beta_{\mathrm{nt}}\,\mathrm{nTrans}_{i, s} + \beta_{\mathrm{st}}\,\mathrm{SameTeam}_{i, s} + \beta_{\mathrm{sn}}\,\mathrm{SameNat}_{i, c}.
\end{split}
\end{equation}
In the above, $A_{is}, H_{is}, W_{is}$ are the age, height and weight, respectively, for the $i^{th}$ player during the season $s$. We consider the quadratic term in the effect of age to match with the idea of diminishing returns, which has been argued in related research works \citep[see][among others]{dendir2016soccer}. The position or role of the player (goalkeeper, defender, midfielder, forward, with goalkeeper serving as the baseline category) is incorporated in the equation via indicators $\ind\{P_i=\cdot\}$ and $\bm{\beta}_{\mathrm{pos}}$ is the associated parameter vector. Next, $R_{i c_{s-1},s-1}$ reflects the rating of the same player in the last season (note that the player may have been in a different club or league). We use $\mathrm{TR}_{c,s}$ and $\mathrm{TR}^{(\mathrm{pos})}_{c,s}$ to denote, respectively, the median rating achieved by all players in the team $c$ and the same by the players in the same position in team $c$. These two features are essentially used to indicate the strength of the overall squad and the strength of the players in the same position. Along the same lines, $\mathrm{TD}^{(\mathrm{pos})}_{c,s}$ is the number of players in the particular position, which reflects the squad depth in the specific role. The variable $\mathrm{nTrans}_{i,s}$ counts the total number of times the player has been transferred before season $s$; $\mathrm{SameTeam}_{i,s}$ is a binary variable indicating if $c_{s-1} = c_s$; and $\mathrm{SameNat}_{i,c}$ is another binary variable denoting if the player's nationality matches with the country the club is in.

Overall, the model can be written in matrix form as
\begin{equation*}
    \bm{R} = \X\bm{\beta} + \mathbf{V}\bm{\eta} + \bm{\varepsilon}, \quad \bm\eta \sim \gauss\left(\bm{0}, \Omega\right),   \;\bm{\varepsilon} \sim 
    \gauss\left(\bm{0},\sigma^2\I\right),
\end{equation*}
where $\bm{R}$ is the vector of all ratings, $\bm{\beta}$ is the vector of parameters mentioned in \eqref{eq:mean-structure-rating} and $\X$ is the corresponding design matrix. The term $\mathbf{V}\bm{\eta}$ is the collection of all random effects with $\Omega$ being a block diagonal matrix with the blocks $\sigma^2_{\mathrm{club}}\I$, $\sigma^2_{\mathrm{cur}}\I$, $\sigma^2_{\mathrm{last}}\I$.

In terms of implementation, the unknown parameters $(\bm{\beta},\sigma^2,\sigma^2_{\mathrm{club}},\sigma^2_{\mathrm{cur}},\sigma^2_{\mathrm{last}})$ are estimated by restricted maximum likelihood (REML) approach, supported by the \texttt{lmerTest} package in \texttt{R} \citep{lmertest}. Once the model is estimated using the training data, we need to obtain the out-of-sample predictions for the values of all players in the next season, if they play for the focal club $c_0$.
To that end, for every candidate player $i$ considered at $c_0$ in season $s+1$, first we form the feature vector $\x_{ic_0,s+1}$ (structured according to the specification of $\X$) based on all information available at the end of the season $s$. Then, the predicted value is
\begin{equation*}
    \hat{R}_{i c_0,\, s+1} = \x_{ic_0,s+1}^{\top}\hat{\bm{\beta}} + \hat{u}_{(c_s\to c_0)} + \hat{u}^{\mathrm{cur}}_{\ell(c_0)} + \hat{u}^{\mathrm{last}}_{\ell(c_s)},
\end{equation*}
where $c_s$ is the team where the player played in the season $s$. It is important to point out that if the corridor $(c_s\to c_0)$ is unobserved historically, we set $\hat{u}_{(c_s\to c_0)}=0$ and use the marginal forecast $\x_{ic_0,s+1}^{\top}\hat{\bm{\beta}}+\hat{u}^{\mathrm{cur}}_{\ell(c_0)}+\hat{u}^{\mathrm{last}}_{\ell(c_s)}$ to ensure coherent predictions for new corridors.

\subsection{Modelling the transfer fee}\label{sec:model-transfer}

The second step of our methodology is to develop a suitable model for the transfer fee of the players. Recall that for any possible move of player $i$ from a selling club $c$ to a buying club $c'$ after the completion of season $s$, we denote the transfer fee by $Y_{i,c\to c',s}$. Since fees are positive and markedly right-skewed, we work with the log-transformed response $Z_{i,c\to c',s}=\log Y_{i,c\to c',s}$ and adopt a linear mixed-effects modelling strategy that (i) captures systematic effects of player attributes, form, market conditions; and (ii) absorbs latent heterogeneity in willingness-to-pay and mark-ups via club-level random effects. Specifically, conditional on observed features for transactions following season $s$, we write
\begin{equation*}
    Z_{i,c\to c',s}  =  \nu_{i,c\to c',s}  +  b^{\mathrm{buy}}_{c'}  +  b^{\mathrm{sell}}_{c}  +  \epsilon_{i,c\to c',s},
\end{equation*}
where $\nu_{i,c\to c',s}$ is the mean structure detailed below, $\epsilon_{i,c\to c',s}$ is white noise, $b^{\mathrm{buy}}_{c'}$ and $b^{\mathrm{sell}}_{c}$ are independent random intercepts corresponding to the buying and the selling clubs respectively. We let
\begin{equation*}
    b^{\mathrm{buy}}_{c'} \sim \gauss\left(0,\sigma^2_{\mathrm{buy}}\right), \qquad
    b^{\mathrm{sell}}_{c} \sim \gauss\left(0,\sigma^2_{\mathrm{sell}}\right), \qquad
    \epsilon_{i,c\to c',s} \sim \gauss\left(0,\tau^2\right).
\end{equation*}
On the other hand, the average log-fee is expressed additively as
\begin{equation}\label{eq:mean-structure-fee}
\begin{split}
    \nu_{i,\,c\to c',\,s}  &= \theta_0  +  \theta_1 t_s  +  \left(\theta_{a1}A_{is} + \theta_{a2}A_{is}^{2}\right) +  \theta_{h}H_{is} + \theta_{w}W_{is} + \bm{\theta}_{\mathrm{pos}}^{\top}\ind\{P_i=\cdot\} + \theta_{\mathrm{cr}} \mathrm{CR}_{is} + \theta_{r}R_{ics} \\
    &\quad + \theta_{\mathrm{tp}} \mathrm{TP}_{is} + \theta_{\mathrm{g}} \mathrm{G}_{is} + \theta_{\mathrm{gc}} \mathrm{GC}_{is}  + \theta_{\mathrm{pen}}\,\mathrm{Pen}_{is} + \theta_{\mathrm{sh}}\,\mathrm{Sh}_{is} + \theta_{\mathrm{pa}} \mathrm{PA}_{is} + \theta_{\mathrm{c}} \mathrm{C}_{is} + \theta_{\mathrm{cl}} \mathrm{Cl}_{is} \\
    & \quad + \theta_{\mathrm{in}}\,\mathrm{Int}_{is} +  \theta_{\mathrm{fs}} \mathrm{Fee}_{\ell(c),s} + \theta_{\mathrm{fb}} \mathrm{Fee}_{\ell(c'),s} + \theta_{\mathrm{td,s}} \mathrm{TD}_{c,s} + \theta_{\mathrm{td,b}} \mathrm{TD}_{c',s}  \\
    &\quad + \theta_{\mathrm{tr,sp,s}} \mathrm{TR}^{\mathrm{(pos)}}_{c,s} + \theta_{\mathrm{tr,sp,b}} \mathrm{TR}^{\mathrm{(pos)}}_{c',s} + \theta_{\mathrm{tr,s}} \mathrm{TR}_{c,s} + \theta_{\mathrm{tr,b}} \mathrm{TR}_{c',s}.
\end{split}
\end{equation}

In the above equation, $t_{s}$ is used as a deterministic time index to capture a linear inflation pattern in the transfer market. The effects of age, height, weight, and position are incorporated in the same way as in model \eqref{eq:mean-structure-rating}, whereas the impact of the player's overall quality is captured through average career rating until the end of season $s$ (denoted as $\mathrm{CR}_{is}$). Next, considering that different aspects of the player's form in the recent season may impact the transfer fee, we include in the model the most recent rating $R_{ics}$, together with total game-time ($\mathrm{TP}_{is}$), goals ($\mathrm{G}_{is}$), goal contribution (total goals and assists, denoted by $\mathrm{GC}_{is}$), penalty-scoring accuracy ($\mathrm{Pen}_{is}$), number of shots ($\mathrm{Sh}_{is}$), passing accuracy ($\mathrm{PA}_{is}$), number of cards (red cards are counted as two yellow cards and the total counts are denoted by $\mathrm{C}_{is}$), number of clearances ($\mathrm{Cl}_{is}$), and the number of interceptions ($\mathrm{Int}_{is}$). The next set of variables in the model are used to capture the effects of market-level conditions at both clubs. The league-season median of realised fees in the seller and buyer leagues are indicated by $\mathrm{Fee}_{\ell(c),s}$ and $\mathrm{Fee}_{\ell(c'),s}$ respectively, where $\ell(c)$ denotes the league in which the club $c$ plays. Finally, the model captures the effects of team depth and overall quality at both clubs using variables akin to model \eqref{eq:mean-structure-rating}. 

Overall, the model can be written in matrix form as
\begin{equation}\label{eq:transfer-model}
    \bm{Z}  =  \W\bm{\theta} + \mathbf{S}\bm{\zeta} + \bm{\epsilon}, 
    \qquad \bm{\zeta} \sim \gauss \left(\bm{0}, \Xi\right), 
    \quad \bm{\epsilon}\sim\gauss\left(\bm{0},\tau^2\I\right),
\end{equation}
where $\bm{Z}$ stacks the observed log-fees, $\W$ is the design matrix obtained from \eqref{eq:mean-structure-fee} with associated parameter vector $\bm{\theta}$, and the term $\mathbf{S}\bm{\zeta}$ collects the buyer/seller random intercepts. Note that the covariance matrix for the random terms is given by a block diagonal matrix with blocks $\sigma^2_{\mathrm{buy}}\I$ and $\sigma^2_{\mathrm{sell}}\I$. We estimate these unknown parameters $(\bm{\theta},\tau^2,\sigma^2_{\mathrm{buy}},\sigma^2_{\mathrm{sell}})$ by REML, as before. Once the model is estimated on historical transfers, we use it to produce out-of-sample forecasts for any prospective move of player $i$ from a current club $c$ to a focal club $c_0$ in the next window. Let $\w_{i,c\to c_0,s}$ denote the corresponding feature vector (structured to match $\W$). The predicted log-fee is then given by
\begin{equation*}
    \widehat{Z}_{i,c\to c_0,s}  =  \w_{i,c\to c_0,s}^{\top}\hat{\bm{\theta}}  +  \hat{b}^{\mathrm{buy}}_{c_0}  +  \hat{b}^{\mathrm{sell}}_{c}.
\end{equation*}

Of course, if either the buying or the selling club has no support in the training data, we set the corresponding estimates to $0$ and use the marginal forecast. Further, to obtain a fee on the original scale, the predicted values are transformed using the exponential function. These values, along with the predicted ratings from \Cref{sec:model-rating}, are used in building the transfer strategy, as elaborated next.

\subsection{Optimisation strategy for transfer}\label{sec:optimisation}

Recall that $\cP$ denotes the entire pool of players and $\cP_{c_0}$ denotes the current squad for a particular club $c_0 \in \mathcal{C}$. In order to design the algorithm of finding the optimal transfer strategy for $c_0$, we primarily rely on two variables: (i) the importance of every player in the proposed squad in the next season, captured through the estimated rating variable $R_i$, and (ii) the predicted transfer fee $Y_i$ for any player if they get transferred to $c_0$. Note that we drop the hat signs and the subscripts indicating clubs and seasons for notational clarity throughout this discussion. As described in the previous section, $Y_i$ is modelled as a log-normal random variable, and we let $\log Y_i \sim \gauss(\mu_i, \sigma_i^2)$. Further, set $r_i$ as the maximum possible selling price for a player who is in the current squad $\cP_{c_0}$ but is going to be excluded from the final squad. Our key objective is to determine the set of decision variables $x_i \in \{0, 1\}$, indicating whether player $i \in \cP$ should be selected for the proposed squad ($x_i = 1$) or not ($x_i = 0$). For this analysis, the total budget available for transfers is assumed to be capped at $B_{\text{max}}$.

At first, we determine the total cost in a transfer window as 
\begin{equation*}
    \mathrm{cost} = \sum_{i \in \cP \setminus \cP_{c_0}} x_i \mathbb{E}\left(Y_i\right) + \sum_{i \in \cP_{c_0}} (1 - x_i) \left(\mathbb{E}\left(Y_i\right) - r_i\right),
\end{equation*}
where the first component in the model comes from buying players from other clubs, whereas the second component comes from selling players from the current squad. Clearly, when the maximum selling price ($r_i$) of a player is less than their expected market value, the club incurs a sure loss, which can be seen as a transfer cost. Next, to determine the risk of the transfer cost, we consider the variability of the cost incurred due to buying players, as any club ideally have complete control over selling a player and can always deny a prospective buyer if the offered value is not deemed suitable. Under that setting, the risk associated with all possible transfers is determined as
\begin{equation*}
    \mathrm{risk} = \sqrt{\sum_{i \in \cP \setminus \cP_{c_0}} x_i \var (Y_i)}.
\end{equation*}
Finally, it is of primary interest to the club to maximise the overall quality of the players in the club. For that, following the earlier mentioned ideas, we define the metric
\begin{equation*}
    \mathrm{quality} = \sum_{i \in \cP} x_i R_i.
\end{equation*}

We treat the optimisation objective as a decision-support criterion rather than an estimated utility function. We introduce mixing parameters $\lambda_1$, $\lambda_2$, and $\lambda_3$ as user-specified preference weights that allow the framework to represent different strategies. A club pursuing immediate sporting success may assign greater weight to projected squad quality, whereas a club operating under tighter financial constraints may assign greater weight to expected cost and risk. The role of the model is therefore conditional: given predicted player ratings, predicted transfer-fee distributions, and a user-specified strategic objective, it identifies the optimum squad configurations. Similarly, several constraints are deliberately specified below as configurable planning parameters rather than fixed assumptions, and the numerical values used in our empirical application are illustrative choices to reflect plausible squad management tactics. Then, with suitable mixing parameters $\lambda_1$, $\lambda_2$ and $\lambda_3$, we define the objective
\begin{equation*}
    \F = -\left(\lambda_1 \; \mathrm{cost} + \lambda_2 \; \mathrm{risk}\right) + \lambda_3 \; \mathrm{quality}
\end{equation*}
which needs to be maximised to ensure that the team's total strength is maximised while minimising the total expected transfer cost and the total risk associated with the transfers (of course, the total cost should be maintained within the maximum budget with high confidence). Additionally, to mimic a realistic transfer scenario, we impose a few other constraints on the variables in the model, motivated by the settings of the game as well as the leagues. These are mentioned below.

\paragraph{Chance constraint}

The price distributions of players are modelled through a log-linear equation. Therefore, to consider the stochasticity of the problem, we impose a chance constraint on the total buying price. With $\alpha>0$ denoting a user-defined threshold for maximum error, this constraint is of the form
\begin{equation*}
    \Pr\left( \sum_{i \in \cP \setminus \cP_{c_0}} x_i Y_i  \leqslant B_{\text{max}} \right) \geqslant 1 - \alpha.
\end{equation*}
To obtain the mean and variance of the total buying cost, we utilise the moment matching approximation formula given by \cite{marlow1967} for the sum of log-normal variables. According to it, the distribution of $\sum_{i \in \cP \setminus \cP_{c_0}} x_i Y_i$ can be approximated with another log-normal distribution whose parameters are given by
\begingroup
\footnotesize
\begin{equation*}
     \mu_{\ast} = \log \left[\sum_{i \in \cP \setminus \cP_{c_0}} x_i\exp \left(\mu_i + \sigma^2_i/2\right)\right] -  \frac{\sigma^2_{\ast}}{2}, \; \;  
     \sigma^2_{\ast} = \log \left[\frac{\sum_{i \in \cP \setminus \cP_{c_0}} \left \{x_i\exp \left(2\mu_i+\sigma^2_i\right) \left(\exp(\sigma^2_i) - 1\right)\right\}} {\left(\sum_{i \in \cP \setminus \cP_{c_0}} x_i\exp\left(\mu_i + \sigma^2_i/2\right) \right)^2} + 1\right].
\end{equation*}
\endgroup
Using the above, and taking $z_{\alpha}$ as the critical value of the standard normal distribution corresponding to the confidence level $(1 - \alpha)$, an alternative representation of the chance constraint can be written as 
\begin{equation}
    \mu_{\ast} + z_{\alpha} \sigma_{\ast} \leqslant \log(B_{\text{max}}). \label{eq:deterministic_chance_constraint}
\end{equation}

\paragraph{Squad specific constraints}

In football, making the squad strength abnormally high defeats the purpose of team selection, as ultimately only the playing 11 and the substitutes participate in the game, and the rest will still have to be paid their salary. The same can be said about restrictions on the number of players in the squad for the different playing positions, which are again based on the dynamics of the game as observed in reality. With that in view, the maximum number of players that can be included in the squad is set to $\overline{k_{\mathrm{tot}}}$, i.e., we take $\sum_{i \in \cP} x_i \leqslant \overline{k_{\mathrm{tot}}}$.   

On the other hand, our dataset $\D$ lists four key positions of the players: Goalkeeper, Defender, Midfielder and Forward. We impose constraints that the minimum number of players in these positions are $\underline{k_{\mathrm{g}}}$, $\underline{k_{\mathrm{d}}}$, $\underline{k_\mathrm{m}}$ and $\underline{k_\mathrm{f}}$, respectively. Moreover, since only one goalkeeper plays at any point in time and they are specialised for that single role, having too many goalkeepers in the squad creates an imbalance. Therefore, we also put a constraint on the maximum number of goalkeepers (denote it as $\overline{k_{\mathrm{g}}}$). For our implementation purposes in this study, we shall take $\overline{k_{\mathrm{tot}}}$ to be 30, $\underline{k_\mathrm{g}}$ is set as $2$,  $\overline{k_\mathrm{g}}$ is $4$, $\underline{k_\mathrm{f}}$ is $4$, $\underline{k_\mathrm{d}}$ and $\underline{k_\mathrm{m}}$ are both $8$. On a related note, during a transfer window, a club may have to buy new position-specific players to meet particular needs, which may occur due to injuries or other unavoidable reasons. To achieve that, we allow a minimum number of position-specific players to be bought during a transfer window and let them be denoted as $\underline{kb_\mathrm{g}}$, $\underline{kb_\mathrm{d}}$, $\underline{kb_\mathrm{m}}$ and $\underline{kb_\mathrm{f}}$, respectively. Since these are specific to every club, for the generality of our implementation, we set these numbers to zero.

Furthermore, during the transfer window, a club may want to build a young and prospective team to ensure a long-term impact as a team while improving the squad's importance. However, to ensure a smooth transition, it is also recommended to retain a certain number of players within the team to match their playing style. So, we impose three constraints related to the quality of the squad: the average age of the squad must not increase, the average value of the squad must not get worse, and the club must retain a minimum of $\underline{k_{\mathrm{retain}}}$ players. We shall set $\underline{k_{\mathrm{retain}}} = 15$ for our illustration.

\paragraph{Brand exposure constraints}

European clubs usually aim to improve their overall brand value in various ways. One such approach is to transfer players between different leagues. To incorporate that, we set a few conditions. First, the minimum number of players to be bought from other continents is denoted as $\underline{kb_{\mathrm{oth}}}$. However, to avoid potentially problematic or expensive bureaucratic issues in buying players from other continents, we also set $\overline{kb_{\mathrm{oth}}}$ to be the maximum number of players bought from other continents. Second, the minimum number of players from the leagues in the same country is set as $\underline{kb_{\mathrm{loc}}}$. And third, the minimum number of transfers from other top European leagues is designated as $\underline{kb_{\mathrm{top}}}$. For illustration purposes in this paper, we shall set $\underline{kb_{\mathrm{top}}}, \; \underline{kb_{\mathrm{oth}}}, \; \underline{kb_{\mathrm{loc}}} = 0$ and $\overline{kb_{\mathrm{oth}}}=2$.

\paragraph{Financial constraints}

We have already noted above that the overall cost must be maintained within the specified budget with high confidence. While that ensures a proper management of the total cost, a club may also want to ensure a certain amount of profit, possibly by selling some players from current squads during the transfer window. We implement this through a constraint so that the minimum profit from selling players is more than $\underline{\mathrm{profit}}$ in the respective currency. Finally, we emphasise that buying or selling too many players can lead to financial fair-play violations discussed in \Cref{sec:contribution}. Therefore, it is important to set a threshold for the total number of transfers ($\overline{k_{\mathrm{transfer}}}$). For our illustration purpose, we set this at 10. 

\paragraph{Additional conditions}

In addition to the aforementioned constraints, we consider four other aspects. First, a club may want to build their gameplay strategy around specific players. Second, a club may want to sell specific players in a transfer window to ensure they do not become free agents. Third, a club may not want to sell their top-performing players unless they are explicitly made available for transfer due to other reasons. And fourth, a club may not want to sell young players as they want to develop the players for a long-term plan. We acknowledge the importance of these four requirements and note that these need not be added as constraints inside the optimisation algorithm, but can be easily managed by appropriately modifying the set of decision variables before running the optimisation routine and by updating the other constraints accordingly. As an example, our illustrations of the proposed algorithm adds a no-sale condition on the topmost goalkeeper, two top-performing defenders, two top-performing midfielders, and the top-performing striker, barring the players who are transfer-listed according to the dataset. Indeed, for such players, it will be interesting to see whether our algorithm suggests selling these players, and if so, what the fee should be. 

Now, for the constrained optimisation, we start by substituting \eqref{eq:deterministic_chance_constraint} into the original problem, and we obtain the deterministic equivalent of the chance constrained optimisation, which, along with the mentioned constraints, can be written in the following manner:
\begin{equation*}
    \text{Maximise} \quad \bigl[-\left(\lambda_1 \; \mathrm{cost} + \lambda_2 \; \mathrm{risk}\right) + \lambda_3 \; \mathrm{quality}\bigr]  
\end{equation*}
subject to decision variables $x_i \in \{0,1\}$ for $i\in \cP$, and the constraints
\begingroup
\small
\setstretch{0.8}
\setlength{\jot}{1pt}
\renewcommand{\theequation}{C\arabic{equation}}
\setcounter{equation}{0}
\newcommand{\csum}[1]{\sum\nolimits_{#1}}
\begin{align}
\mu_{\ast} + z_{\alpha} \sigma_{\ast} &\leqslant \log(B_{\text{max}}), \label{con:C1}\\
\csum{i\in \cP} x_i &\leqslant \overline{k_{\mathrm{tot}}}, \label{con:C2}\\
\csum{i\in \cP_{c_0}} x_i &\geqslant \underline{k_{\mathrm{retain}}}, \label{con:C3}\\
\csum{i \in \cP_{c_0}} (1-x_i) + \csum{i \in \cP \setminus \cP_{c_0}} x_i
&\leqslant \overline{k_{\mathrm{transfer}}}, \label{con:C4}\\
\csum{i \in \cP_{c_0}} (1-x_i)\,\mathbb{E}(Y_i) &\geqslant \underline{\mathrm{profit}}, \label{con:C5}\\
\underline{k_{\mathrm{g}}}\leqslant \csum{i\in \cP} x_i\,\mathbb{I}(P_i=\text{goalkeeper})
&\leqslant \overline{k_{\mathrm{g}}}, \label{con:C6}\\
\csum{i\in \cP} x_i\,\mathbb{I}(P_i=\text{defender}) &\geqslant \underline{k_{\mathrm{d}}}, \label{con:C7}\\
\csum{i\in \cP} x_i\,\mathbb{I}(P_i=\text{midfielder}) &\geqslant \underline{k_{\mathrm{m}}}, \label{con:C8}\\
\csum{i\in \cP} x_i\,\mathbb{I}(P_i=\text{forward}) &\geqslant \underline{k_{\mathrm{f}}}, \label{con:C9}\\
\csum{i\in \cP\setminus \cP_{c_0}} x_i\,\mathbb{I}(P_i=\text{goalkeeper}) &\geqslant \underline{kb_{\mathrm{g}}}, \label{con:C10}\\
\csum{i\in \cP\setminus \cP_{c_0}} x_i\,\mathbb{I}(P_i=\text{defender}) &\geqslant \underline{kb_{\mathrm{d}}}, \label{con:C11}\\
\csum{i\in \cP\setminus \cP_{c_0}} x_i\,\mathbb{I}(P_i=\text{midfielder}) &\geqslant \underline{kb_{\mathrm{m}}}, \label{con:C12}\\
\csum{i\in \cP\setminus \cP_{c_0}} x_i\,\mathbb{I}(P_i=\text{forward}) &\geqslant \underline{kb_{\mathrm{f}}}, \label{con:C13}\\
\underline{kb_{\mathrm{oth}}}\leqslant
\csum{i\in \cP\setminus \cP_{c_0}} x_i\,\mathbb{I}(i\in \text{other continents})
&\leqslant \overline{kb_{\mathrm{oth}}}, \label{con:C14}\\
\csum{i\in \cP\setminus \cP_{c_0}} x_i\,\mathbb{I}(i\in \text{other top leagues})
&\geqslant \underline{kb_{\mathrm{top}}}, \label{con:C15}\\
\csum{i\in \cP\setminus \cP_{c_0}} x_i\,\mathbb{I}(i\in \text{same country})
&\geqslant \underline{kb_{\mathrm{loc}}}, \label{con:C16}\\
\frac{\csum{i\in \cP} x_i A_i}{\csum{i\in \cP} x_i}
&\leqslant \frac{\csum{i\in \cP_{c_0}} A_i}{\abs{\cP_{c_0}}}, \label{con:C17}\\
\frac{\csum{i\in \cP} x_i R_i}{\csum{i\in \cP} x_i}
&\geqslant \frac{\csum{i\in \cP_{c_0}} R_i}{\abs{\cP_{c_0}}}. \label{con:C18}
\end{align}
\endgroup

In the above, with a slight abuse of notation, we use $P_i$ as a generic symbol to denote the players. Note that the last two constraints are used to ensure that the average age of the squad does not increase, and that the average value does not get worse. Now, to solve this constrained optimisation problem, we employ a genetic algorithm-based solver. Specifically, we construct a fitness function 
\begin{equation}\label{eq:multi_objective_criteria}
	\F_{\beta}(x) = \bigl[-\left(\lambda_1 \; \mathrm{cost} + \lambda_2 \; \mathrm{risk}\right) + \lambda_3 \; \mathrm{quality}\bigr]  - \beta \sum_{k=1}^{K} \max\left(0, \text{Constraint}_k\right),
\end{equation}
where $\beta \gg 0$ is a penalty coefficient and $\text{Constraint}_k$ denotes the functional form of the $k^{th}$ individual constraint mentioned in \eqref{con:C1} to \eqref{con:C18} when written in the form $\text{Constraint}_k \leqslant 0$. This ensures that the maximum is not achieved when a constraint is violated. 

A step-by-step summary of our optimisation strategy is provided below, in \Cref{alg:transfer:opt}. 

\begingroup
\begin{algorithm}[h]
\caption{Transfer strategy}
\label{alg:transfer:opt}
\begin{small}
\setstretch{0.85}
\begin{algorithmic}[1]
\State Initiate $\cP$ as the set of all available players and $\cP_{c_0}$ as the set of current players
\State Filter must-buy players outside the decision set
\Statex \hspace{\algorithmicindent}(a) $\cP \gets \cP \setminus [\text{Must Buy}]$
\Statex \hspace{\algorithmicindent}(b) Update all constraint bounds as appropriate
\State Filter must-sell players outside the decision set
\Statex \hspace{\algorithmicindent}(a) $\cP \gets \cP \setminus [\text{Must Sell}]$
\Statex \hspace{\algorithmicindent}(b) $\cP_{c_0} \gets \cP_{c_0} \setminus [\text{Must Sell}]$
\Statex \hspace{\algorithmicindent}(c) Update the minimum-profit constraint and the maximum number of transfers

\State Filter must-keep players outside the decision set
\Statex \hspace{\algorithmicindent}(a) $\cP \gets \cP \setminus [\text{Keep Top Players}] \setminus [\text{Keep Young Players}]$
\Statex \hspace{\algorithmicindent}(b) $\cP_{c_0} \gets \cP_{c_0} \setminus [\text{Keep Top Players}] \setminus [\text{Keep Young Players}]$
\Statex \hspace{\algorithmicindent}(c) Update relevant constraint bounds on squad composition

\State Choose scaling parameters $\lambda_1,\lambda_2,\lambda_3$  and construct the objective function
\State Select a sufficiently large penalty parameter $\beta$ to construct the fitness function $\F_{\beta}(x)$

\State Employ the heuristic optimisation algorithm
\Statex \hspace{\algorithmicindent}(a) Run the heuristic until convergence or until no improvement in feasible fitness is observed
\Statex \hspace{\algorithmicindent}(b) Check the best candidate solution against all original constraints
\Statex \hspace{\algorithmicindent}(c) If infeasible, update convergence settings and rerun

\State Report the best feasible solution found
\end{algorithmic}
\end{small}
\end{algorithm}
\endgroup

In particular, the optimisation problem combines binary decisions with nonlinear financial-risk terms and several rule-based constraints. The chance constraint depends on the distribution of the sum of log-normal transfer costs, the objective includes a nonlinear risk component, and the feasibility set includes appropriate restrictions. Although restricted versions of the model could be approximated and solved using mixed-integer programming or conic optimisation, doing so would require additional modelling choices for the stochastic budget constraint and nonlinear terms. We therefore use a flexible heuristic solver that retains the original decision-support structure and accommodates practical constraints such as must-buy, must-sell, no-sale, positional, and nationality restrictions without substantial reformulation. For selecting the solver, we benchmark between genetic algorithm, simulated annealing, and random restart hill climbing, which we present in the supplementary material (refer to Section S4). Note that the choice of the algorithm is not essential to the modelling framework and is used here as an implementation device rather than as the primary methodological contribution. The penalty term in \eqref{eq:multi_objective_criteria} is used only to guide the search; after convergence, each candidate solution is evaluated against the original constraints, and only feasible solutions are retained. If the best candidate violates any hard constraint, the maximum iteration is increased, and the algorithm is rerun once. If the rerun returns an infeasible solution, we stop our analysis to save time.

\subsection{Competitive bidding framework}
\label{sec:bidding}
As an extension to obtaining the optimised transfer decisions, we wish to understand the market behaviour when multiple clubs are interested in buying the same player. One should note that the valuation of the player can vary for different clubs based on their financial strength. As a result, the market value of the player can increase when a dominant club is interested in buying a player. This motivates us to formulate this problem as a first-price auction with asymmetric bidders \citep{lebrun1999}.

Consider a player $i$ to be sold by a club $c_0$ and let $\mathcal{C}^i$ denote the set of potential clubs interested in buying the player. We assume that each club $c$ has an independent private valuation $S^i_c>0$ for the player, with $\log S^i_c \sim \mathcal{N}\left(\mu_{ic},\sigma_{ic}^2\right)$ for each $c\in\mathcal{C}^i$ independently. Now, if club $c$ pays an amount $b$ for the player against a valuation $s$, its surplus is equal to $s-b$. Clearly, each club would aim to adopt a strictly increasing bidding strategy $\kappa_c(s) = b$ that maximises the surplus. However, in practical situations, the seller club $c_0$ may not have a deterministic valuation of a player. In such cases, taking inspiration from \cite{li2003}, we consider a distribution for the reserve price $\rho^i$ such that $\log \rho^i \sim \mathcal{N}\left(\mu_{ic_0},\sigma_{ic_0}^2\right)$. For notational convenience, we drop the player index $i$ from now on (except for $\mathcal{C}^i$) and denote the cumulative distribution function associated with the random reserve by $H$. Further, introduce a bid-dependent acceptance through a function $p_c(b)\in[0,1]$, representing the probability that the seller accepts a bid from a club $c$. In practice, it has a resemblance to the auction mechanism with asymmetric reserve prices \citep{kotowski2018}. However, \citet{kotowski2018} considered deterministic asymmetric reserve prices for buyers, i.e., instead of using a distribution function $H$, they considered different fixed reserve prices for different buyers. In football, this may not be ideal as a seller club often sells a player at a lower price than their expected fee, which will not be possible if there is a fixed reserve price. Therefore, it is more flexible to have a random reserve price. On the other hand, a club-specific acceptance rule also ensures that direct rivals have less chance of buying a player, or alternatively, they have to pay significantly more than other buyers. Moreover, such an acceptance mechanism also allows a no-sale outcome after negotiation that mimics realistic football transfer scenarios. 

Additionally, we truncate the distribution function of the random reserve price and the acceptance probability at a fixed bid $\upsilon_{\mathsf{thresh}}$. This is particularly important as in football transfers there is a notion of a buyout clause which states that a player has to be sold by the club if another club is ready to match a predetermined amount, and in such situations the seller's preference will not matter. So, in practice, $\upsilon_{\mathsf{thresh}}$ acts as the amount in the buyout clause such that
\begin{align*}
    \hat{H}(b) = \min\left\{\frac{H(b)}{H(\upsilon_{\mathsf{thresh}})}, 1\right\}, \qquad \hat{p}_c(b) = \min\left\{\frac{p_c(b)}{p_c(\upsilon_{\mathsf{thresh}})}, 1\right\}.
\end{align*}
For illustration purposes, we fix this amount to 95\% quantile of the reserve distribution $H$. Now, let $F_c$ and $f_c$ denote the distribution and density functions of $S_c$ and let $\psi_j(\cdot) = \kappa^{-1}_j(\cdot)$. Then for a bid $b$, the probability that club $c$ outbids a club $j\not = c$ is
\begin{equation*}
    \Pr\left(b \geqslant \kappa_j(S_j)\right) = \Pr(\psi_j(b)\geqslant S_j) = F_j(\psi_j(b)).
\end{equation*}
Therefore, under independence assumption, the probability that club $c$ outbids all rivals is $G_c(b)=\prod_{j\neq c}F_j(\psi_j(b))$. Similarly, the probability that bid $b$ exceeds the random reserve price equals $H(b)=\Pr(\rho\leqslant b)$. Hence, the selling probability of the player to club $c$ is
\begin{equation}
q_c(b)=\hat{p}_c(b)\hat{H}(b)G_c(b)
      =\hat{p}_c(b)\hat{H}(b)\prod_{j\neq c}F_j(\psi_j(b)).
\label{eq:alloc-prob}
\end{equation}
With these taken into account, the expected utility function for club $c$ given a valuation $s$ is 
\begin{equation}
U_c(b;s)=(s-b)\,q_c(b).
\label{eq:utility-general}
\end{equation}

To determine the equilibrium bidding behaviour, each club maximises the utility in \eqref{eq:utility-general} taking into account the distribution of rival valuations, reserve distribution and their corresponding bidding strategies. This maximisation occurs at the Bayes--Nash equilibrium, which is obtained by solving the associated first-order conditions of \eqref{eq:utility-general}. This leads to the following two propositions that give us a framework for solving the equilibrium strategy using numerical integration techniques. The proofs of these results are provided in Section S1 of the supplement.

\begin{proposition}
\label{prop:affinity-foc}
Suppose $F$ and $H$ are continuously differentiable with strictly positive densities and $p_c(\cdot)$ and $\psi_c(\cdot)$ are continuously differentiable function on $\left(0, \upsilon_{\mathsf{thresh}}\right]$. Then, in an interior pure-strategy equilibrium under risk neutrality, the inverse bidding functions satisfy
\begin{equation}
\frac{1}{s-b}
=
\frac{p_c'(b)}{p_c(b)}
+
\frac{h(b)}{H(b)}
+
\sum_{j\neq c}\frac{f_j(\psi_j(b))}{F_j(\psi_j(b))}\psi'_j(b) .
\label{eq:foc}
\end{equation}
Moreover, if $p_c(b)\equiv p_c$ is constant, equilibrium bids coincide
with those obtained under $p_c(b)\equiv1$; that is only expected utilities are
scaled by $p_c$.
\end{proposition}

\begin{proposition}
\label{thm:asymm-ode}
Suppose $F$ and $H$ are continuously differentiable with strictly
positive densities on
$(\underline{s},\overline{s}]$,
and $p(\cdot)$ and $\psi_c(\cdot)$ are continuously differentiable in the interval $\left[b_{\min}, \upsilon_{\mathsf{thresh}}\right]$ with $p(b), \psi_c(b)>0$. Then, we get the following system of ordinary differential equations for $c\in\mathcal{C}^i$:
\begin{equation*}
    \psi'_c(b) = \frac{1}{\left(\abs{\mathcal{C}^i} - 1\right)}\frac{F_c(\psi_c(b))}{f_c(\psi_c(b))}\left[\sum_{j\not = c}\left(\frac{1}{\psi_j(b)-b} - \frac{p'_j(b)}{p_j(b)}\right)-
    \left(\abs{\mathcal{C}^i}-2\right)\left(\frac{1}{\psi_c(b)-b} - \frac{p'_c(b)}{p_c(b)} \right) -
    \frac{h(b)}{H(b)}
    \right].
\end{equation*}
\end{proposition}

\begin{remark}[Maskin-Riley ordering]\label{rem:maskin:riley}
\citet{maskin2000a} established that in an asymmetric first-price auction, clubs with stochastically higher valuations ($\mu_c$ larger) bid more at every valuation and shade more in absolute terms: $\kappa_c(s)>\kappa_j(s)$ and $s-\kappa_c(s)>s-\kappa_j(s)$ pointwise when $\mu_c>\mu_j$. This ordering is not directly preserved, as we have an affinity function that affects the strict ordering. However, we can show with numerical integration that if these affinity functions are similar to each other, then Maskin-Riley ordering holds.
\end{remark}

\begin{remark}[Numerical Scheme]\label{rem:num:ode}
Note that the existence of a Bayes--Nash equilibrium also ensures that we can employ numerical algorithms to solve the system of first-order ordinary differential equations. We provide a brief description of such an algorithm in Section S5 of the supplementary material for multiple rounds of bidding, which can easily be used for a single round of negotiation.
\end{remark}

\section{Results and discussions}\label{sec:results}

For our main analysis, we rely on two sources of data. The information on players' ratings and performances is taken from the WhoScored website\footnote{WhoScored: \url{https://www.whoscored.com/}}, whereas the information on players' characteristics and transfer data is obtained from the Transfermarkt website\footnote{Transfermarkt: \url{https://www.transfermarkt.co.uk/}}. Datasets extracted from these two sources are merged, cleaned, and pre-processed appropriately for our analysis. All computations are carried out in \texttt{R}. The codes and the cleaned dataset will be shared publicly in a GitHub repository. The final dataset we use in our analysis contains information on 20403 players from 500 clubs playing in the leagues of 17 different countries from four continents (11 countries from Europe, 2 countries from North America, 2 countries from South America, and 2 countries from Asia), during the seasons 2009/10 to 2019/20. Note that for every country, different divisions' data are included. For more recent seasons or for other countries, detailed information on all variables were not available, and therefore, we restrict our analysis to these 10 years and the 17 countries. In terms of nationality of the players, there are 167 countries, with a number of players ranging between 1 and 7072 (median number of players per country is 52, and the mean is 404.4). The distribution of this is illustrated in \Cref{fig:nationality_counts}, which confirms that the maximum number of players in the dataset are from the United Kingdom and Brazil, closely followed by Argentina, France, Germany, Spain, and Italy. 

\begin{figure}[!ht]
    \centering
    \includegraphics[width=0.6\textwidth,keepaspectratio]{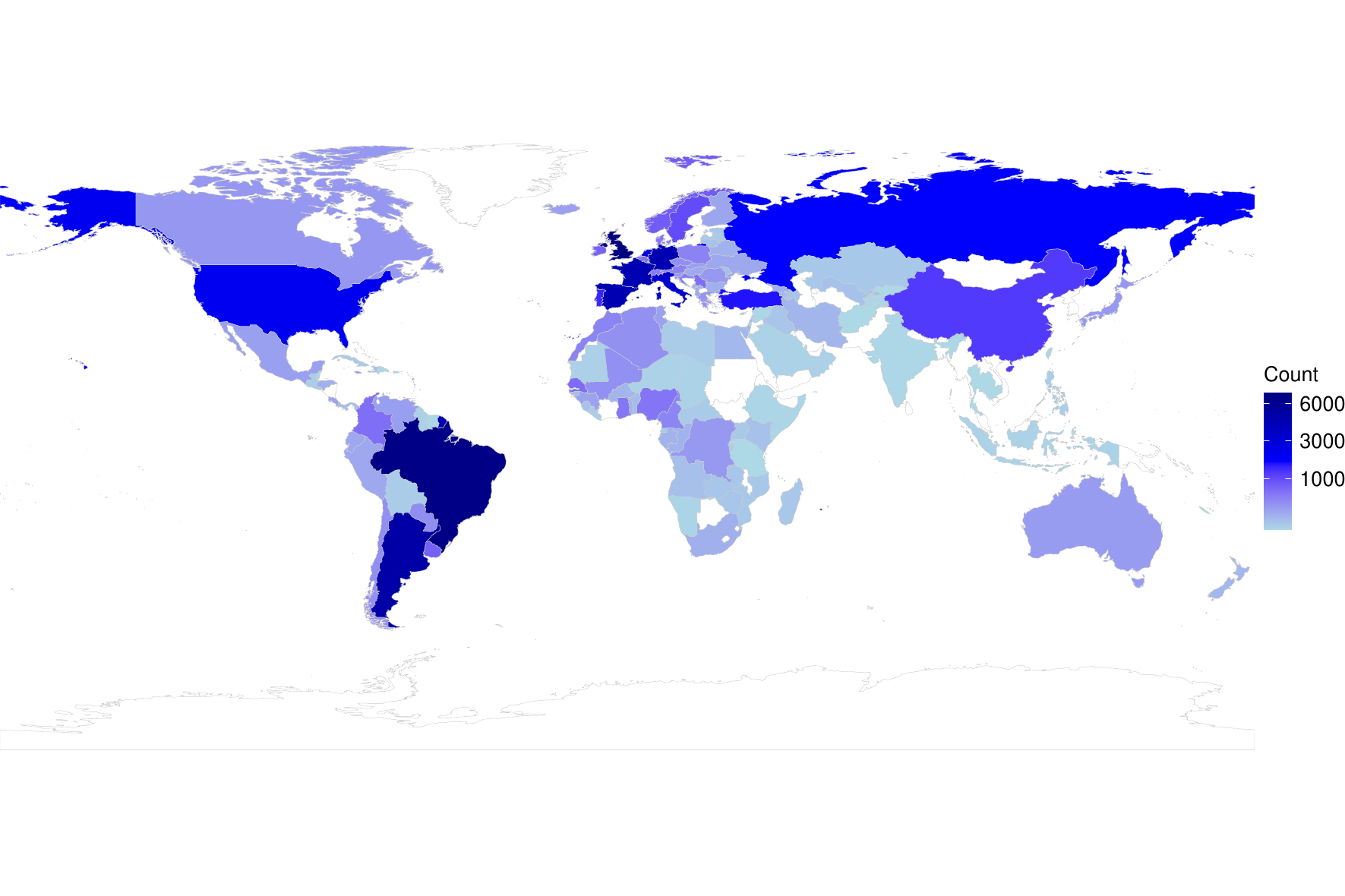}
    \caption{Distribution of players in our dataset according to their nationality.}
    \label{fig:nationality_counts}
\end{figure}

\subsection{Output of the model for player ratings}

The estimates in \Cref{tab:final-model-output:rating} provide a coherent picture of how player ratings evolve from one season to the next, once we control simultaneously for individual attributes, persistence in performance, and the competitive environment in which the player operates, while allowing for unobserved corridor and league heterogeneity through the random-effects structure. 
\begin{table}[ht]
\centering
\caption{Estimated coefficients for the features used in the mixed effects model, for studying the impact on player ratings.}
\label{tab:final-model-output:rating}
{\scriptsize
\begin{tabular}{llcccc}
  \toprule
  Type of feature & Variable & Estimate & Standard error & $t$-score & $p$-value \\
  \midrule
  General & Intercept & $-2.828$$^{***}$ & $0.140$ & $-20.135$ & $0.000$ \\
  \midrule
  Player's & Age (scaled) & $0.570$$^{***}$ & $0.038$ & $15.127$ & $0.000$ \\
  information & Age squared (scaled) & $-0.106$$^{***}$ & $0.007$ & $-15.121$ & $0.000$ \\
  & Position (defender) & $-0.009$ & $0.009$ & $-1.013$ & $0.311$ \\
  & Position (midfielder) & $0.015$ & $0.009$ & $1.655$ & $0.098$ \\
  & Position (forward) & $0.035$$^{***}$ & $0.007$ & $4.738$ & $0.000$ \\
  & Height & $0.024$ & $0.018$ & $1.297$ & $0.195$ \\
  & Weight & $0.002$$^{***}$ & $0.000$ & $7.753$ & $0.000$ \\
  & Last season's rating & $0.303$$^{***}$ & $0.005$ & $65.004$ & $0.000$ \\
  & Number of times transferred & $-0.004$$^{**}$ & $0.001$ & $-2.867$ & $0.004$ \\
  & If same team as last season & $-0.003$ & $0.003$ & $-0.738$ & $0.461$ \\
  & Playing in the birth country & $-0.007$$^{*}$ & $0.003$ & $-2.105$ & $0.035$ \\
  \midrule
  Team's & Overall quality & $0.257$$^{***}$ & $0.021$ & $12.151$ & $0.000$ \\
  information & Quality in same position & $0.723$$^{***}$ & $0.011$ & $67.059$ & $0.000$ \\
  & \#players in same position & $0.004$$^{***}$ & $0.001$ & $5.393$ & $0.000$ \\
  \bottomrule
\end{tabular}}
\end{table}

First, the age effect is found to be strongly nonlinear: the positive linear age term combined with the negative quadratic term indicates an inverted-U relationship between age and predicted rating, consistent with a development phase followed by diminishing returns and eventual decline \citep{dendir2016soccer}. From the standpoint of transfer planning, this curvature is important because it implies that the marginal benefit of recruiting older players depends sharply on where they lie on the age curve; the optimisation step therefore internalises that not all experienced signings are equally valuable, and that prime-age targets yield the largest expected rating uplift per unit cost, all else equal. The position indicators show that forwards are expected to get better ratings than the other three positions due to their higher chances of scoring and directly contributing to decisive match outcomes, which is in line with the arguments proposed by \cite{ball2025comparing}. In the optimisation layer, this motivates the use of position-specific constraints and careful interpretation of cross-position comparisons: selection should be driven by marginal improvements within a role rather than by direct comparison of raw ratings across roles. Anthropometrics play a limited role: height is not strongly informative once the remaining features are included, whereas weight shows a small but robust association with ratings. In applied terms, this implies that body size proxies contribute marginally compared to performance history and team context; consequently, their value in recruitment is likely as secondary screening features rather than primary drivers of selection.

Unsurprisingly, last season's rating is the strongest predictor of next season's rating, indicating substantial persistence in performance after controlling for age, position, and context. At the same time, the relationship is consistent with partial mean reversion. This is useful for transfer planning because it stabilises counterfactual forecasts while discouraging linear extrapolation of peak form, reducing the risk of overvaluing short-lived outliers when balancing expected quality against fee and risk. Transfer history has a small but negative association with subsequent ratings, suggesting that frequent movers underperform as compared to similar peers, plausibly due to adaptation costs or unobserved instability. In contrast, the indicator for staying at the same club adds little once last-season performance and team context are included, implying that continuity is largely captured through recent form and environment rather than providing an independent boost. On the other hand, contrary to popular beliefs, playing in the birth country is found to have a weakly negative impact on performance. Although this finding is in line with earlier research \citep[e.g.,][]{marcen2016bosman}, it should be interpreted cautiously as a conditional association rather than a structural disadvantage of domestic play. It likely reflects selection and composition effects in who becomes an international mover and the leagues or clubs they join. For our main objective, its role is modest: it can shift rankings among similar candidates in cross-border comparisons without driving decisions on its own.

We also find the team context variables to be highly informative, which align with the portfolio nature of squad building. Higher overall team quality is associated with higher individual ratings, consistent with an environment effect. More importantly, stronger same-position units are associated with higher individual ratings, implying complementarities within positional groups: the expected contribution of a signing depends on the quality of the surrounding unit, not only on player attributes. The positive association with the number of players in the same position also points to a depth effect, supporting position-specific depth constraints as performance-relevant rather than purely operational.

In the interest of space, we defer some additional discussions on this model, including a comparison study with other frameworks, to the supplementary material (Section S2). Overall, the results reinforce a key message: expected player contribution is jointly determined by individual trajectory and squad context, thereby justifying why transfer decisions should be optimised at the squad level under constraints, not by ranking players in isolation.  

\subsection{Output of the model for transfer valuation}

Moving on to the second key step of our methodology, we look at the estimated coefficients of the transfer fee model described by equations \eqref{eq:mean-structure-fee} and \eqref{eq:transfer-model}. \Cref{tab:final-model-output:price} shows how the expected transfer fee varies with player attributes, recent form, and market context, after controlling for persistent buyer and seller-club heterogeneity through random intercepts. Since the response is modelled on the log scale, covariates act approximately multiplicatively on fees on the original scale; consequently, statistically meaningful effects translate into economically interpretable percentage shifts in expected fees, which is exactly the quantity required by our optimisation routine. 

\begin{table}[ht]
\centering
\caption{Estimated coefficients for the features used in the mixed effects model, for studying the impact on transfer fees.}
\label{tab:final-model-output:price}
{\scriptsize
\begin{tabular}{llcccc}
  \toprule
  Type of feature & Variable & Estimate & Standard error & $t$-score & $p$-value \\
  \midrule
  General & Intercept & $-17.552$$^{***}$ & $2.156$ & $-8.140$ & $0.000$ \\
  & Linear trend & $0.064$$^{***}$ & $0.012$ & $5.340$ & $0.000$ \\
  \midrule
  Player's & Age (scaled) & $2.361$$^{**}$ & $0.743$ & $3.178$ & $0.002$ \\
  overall & Age squared (scaled) & $-0.668$$^{***}$ & $0.150$ & $-4.459$ & $0.000$ \\
  characteristics & Position (defender) & $-0.323$ & $0.198$ & $-1.635$ & $0.102$ \\
  & Position (midfielder) & $-0.143$ & $0.197$ & $-0.726$ & $0.468$ \\
  & Position (forward) & $0.089$ & $0.167$ & $0.533$ & $0.594$ \\
  & Height & $1.692$$^{**}$ & $0.520$ & $3.251$ & $0.001$ \\
  & Weight & $0.001$ & $0.004$ & $0.181$ & $0.857$ \\
  & Career rating & $1.601$$^{***}$ & $0.129$ & $12.392$ & $0.000$ \\
  \midrule
  Player's & Rating & $-0.373$$^{**}$ & $0.136$ & $-2.745$ & $0.006$ \\
  performance & Game-time & $0.018$$^{***}$ & $0.003$ & $6.036$ & $0.000$ \\
  in last season & Goals & $0.399$ & $0.342$ & $1.166$ & $0.244$ \\
  & Goal contributions & $0.270$ & $0.266$ & $1.015$ & $0.310$ \\
  & Penalty accuracy & $-0.461$ & $0.699$ & $-0.659$ & $0.510$ \\
  & Shots & $0.169$$^{***}$ & $0.043$ & $3.894$ & $0.000$ \\
  & Passing accuracy & $0.023$$^{***}$ & $0.005$ & $4.216$ & $0.000$ \\
  & Cards & $0.108$ & $0.156$ & $0.692$ & $0.489$ \\
  & Clearance & $-0.006$ & $0.020$ & $-0.303$ & $0.762$ \\
  & Interception & $-0.061$ & $0.043$ & $-1.401$ & $0.161$ \\
  \midrule
  Leagues and & Median selling price (seller league) & $0.169$$^{***}$ & $0.019$ & $8.966$ & $0.000$ \\
  teams in & Median buying price (buyer league) & $0.056$$^{***}$ & $0.010$ & $5.665$ & $0.000$ \\
  last season & \#players in same position (seller) & $-0.007$ & $0.014$ & $-0.522$ & $0.602$ \\
  & \#players in same position (buyer) & $-0.003$ & $0.012$ & $-0.219$ & $0.827$ \\
  & Quality in same position (seller) & $-0.127$ & $0.198$ & $-0.644$ & $0.520$ \\
  & Quality in same position (buyer) & $0.140$ & $0.197$ & $0.712$ & $0.477$ \\
  & Overall quality (seller) & $0.748$$^{*}$ & $0.302$ & $2.474$ & $0.013$ \\
  & Overall quality (buyer) & $-0.184$ & $0.198$ & $-0.932$ & $0.351$ \\
  \bottomrule
\end{tabular}}
\end{table}

In the output presented in \Cref{tab:final-model-output:price}, we notice a significant positive linear trend, which confirms systematic fee inflation over time, motivating an explicit time adjustment when forecasting acquisition budgets across windows. Age enters nonlinearly, implying a concave lifecycle pattern: fees increase with age up to a point and then decline, consistent with the trade-off between prime-age readiness and remaining resale horizon \citep[similar findings were confirmed earlier by other researchers, such as][]{depken2021football}. Position effects indicate that forwards command a slight premium relative to the defenders. This finding is in line with the recent study by \cite{metelski2024transfer}, and reinforces the view that fee formation is position-specific and that budget allocation must be planned at the positional level. Among anthropometrics, height has a positive association with fees, whereas weight contributes little once other controls are included.

A persistent quality signal, captured by career rating, is a dominant driver of fees, suggesting that markets price longer-horizon evidence of ability and reputation rather than relying solely on short-term form. The contemporaneous rating enters with a negative partial effect when career rating is included, which should be interpreted as a conditional relationship arising from overlapping information and other market frictions (e.g.\ contract terms and negotiation conditions) rather than as evidence that better recent performance lowers fees. On a related note, among last-season performance statistics, game-time, shots, and passing accuracy are positively associated with fees, indicating that the market puts greater emphasis on overall performance metrics beyond headline outcomes, such as goals and goal contributions, which add a limited incremental signal once other measures are controlled for.

Market context terms are economically important and directly support cross-league planning. The median fee levels in the seller and buyer leagues are positively associated with the transfer fee, capturing league-specific price regimes and purchasing power effects; consequently, comparable player profiles can transact at systematically different price levels depending on the origin and destination leagues. Seller-club strength is also positively associated with fees, consistent with greater bargaining power and weaker sell pressure at stronger clubs, implying that targets sourced from top sellers can be systematically costlier than their on-field profile alone would suggest. More detailed remarks on this model's output, along with comparison against benchmark techniques, are deferred to Section S3 of the supplement in the interest of space. Overall, we may conclude that the model delivers a parsimonious decomposition of fee formation into (i) market inflation and lifecycle pricing, (ii) role and persistent quality premium, (iii) selected stable performance signals, and (iv) league and seller-side bargaining environment, yielding forecast inputs that are well-aligned with the budget and risk components of the subsequent transfer optimisation.

\subsection{Model-recommended transfer strategy}

Following the estimation of player ratings and transfer fees, we use these models to identify high-quality feasible transfer recommendations by solving the optimisation problem described in \cref{sec:optimisation}. Throughout this section, the term ``recommended squad'' refers to the best feasible solution identified by the genetic algorithm under the specified objective function and constraints, interpreted as a near-optimal squad configuration rather than a certified global optimum. This is illustrated for the clubs in the EPL at the end of the 2018/19 season. Specifically, we restrict our analysis to 17 clubs, excluding Cardiff City, Fulham and Huddersfield Town, who were relegated to a lower division at the end of that season. It allows us to compare the recommended strategy with the actual strategy adopted by the clubs and how it led to their performance during the 2019/20 season. For the choice of the maximum transfer budget $B_{\text{max}}$, we assume that a club may spend at most twice its average transfer expenditure over the previous five seasons. The weights $(\lambda_1,\lambda_2,\lambda_3)$ in the multi-criteria formulation \eqref{eq:multi_objective_criteria} are user-defined preference parameters rather than estimated quantities. They are not intended to be universal behavioural parameters; rather, they represent club-specific strategic postures that encode the decision-maker's relative emphasis on expenditure control, financial risk, and squad quality. For the EPL illustration, we specify three broad preference scenarios based on financial capacity. High-budget clubs (above $\text{\euro} 200$ million) are assigned $(\lambda_1,\lambda_2,\lambda_3)=(0.1,0.1,0.8)$, reflecting a strategy that prioritises squad-quality improvement while still accounting for cost and risk. Mid-budget clubs (between $\text{\euro} 100$ million and $\text{\euro} 200$ million) are assigned $(0.2,0.2,0.6)$, representing a more balanced trade-off. Lower-budget clubs (at most $\text{\euro} 100$ million) are assigned $(0.3,0.3,0.4)$, reflecting greater emphasis on cost and risk control. These values should be interpreted as transparent scenario choices used to illustrate the framework across heterogeneous clubs, not as universally optimal parameters. Alternative weights can be used to conduct scenario analysis under different strategic priorities, and the framework can support systematic sensitivity analysis over the weight space. For the minimum profit constraint ($\underline{\mathrm{profit}}$), we consider Transfermarkt data on the average income over the last five seasons. For implementation and illustration in this paper, we use the \texttt{GA} package in \texttt{R} \citep{ga_in_r}, based on our benchmark study presented in Section S4 of the supplementary material.

Before turning to the results, we clarify the validation objective. The empirical exercise is not intended to claim causal superiority over recorded club decisions, since actual transfer outcomes are shaped by private information, contract negotiations, wages, injuries, manager preferences, and availability constraints that are not fully observable in public data. Rather, we evaluate the framework as a public-data decision benchmark and diagnostic plausibility exercise, asking three questions. First, do the counterfactual squads generated by the model produce predicted ratings and transfer costs within realistic ranges? Second, do the recommended buy/sell decisions align with a meaningful subset of observed transfers? Third, when the recommendations differ from observed decisions, are the discrepancies interpretable in terms of omitted contractual, tactical, or strategic factors? This validation logic is appropriate for public-data transfer modelling, where the analyst observes recorded transactions but not the full negotiation set, wage demands, contract clauses, player preferences, or internal club objectives. Accordingly, all recommendations reported below are conditional on the preference weights and planning constraints described above.

\begin{figure}[!ht]
    \centering
    \includegraphics[width=0.77\textwidth,keepaspectratio]{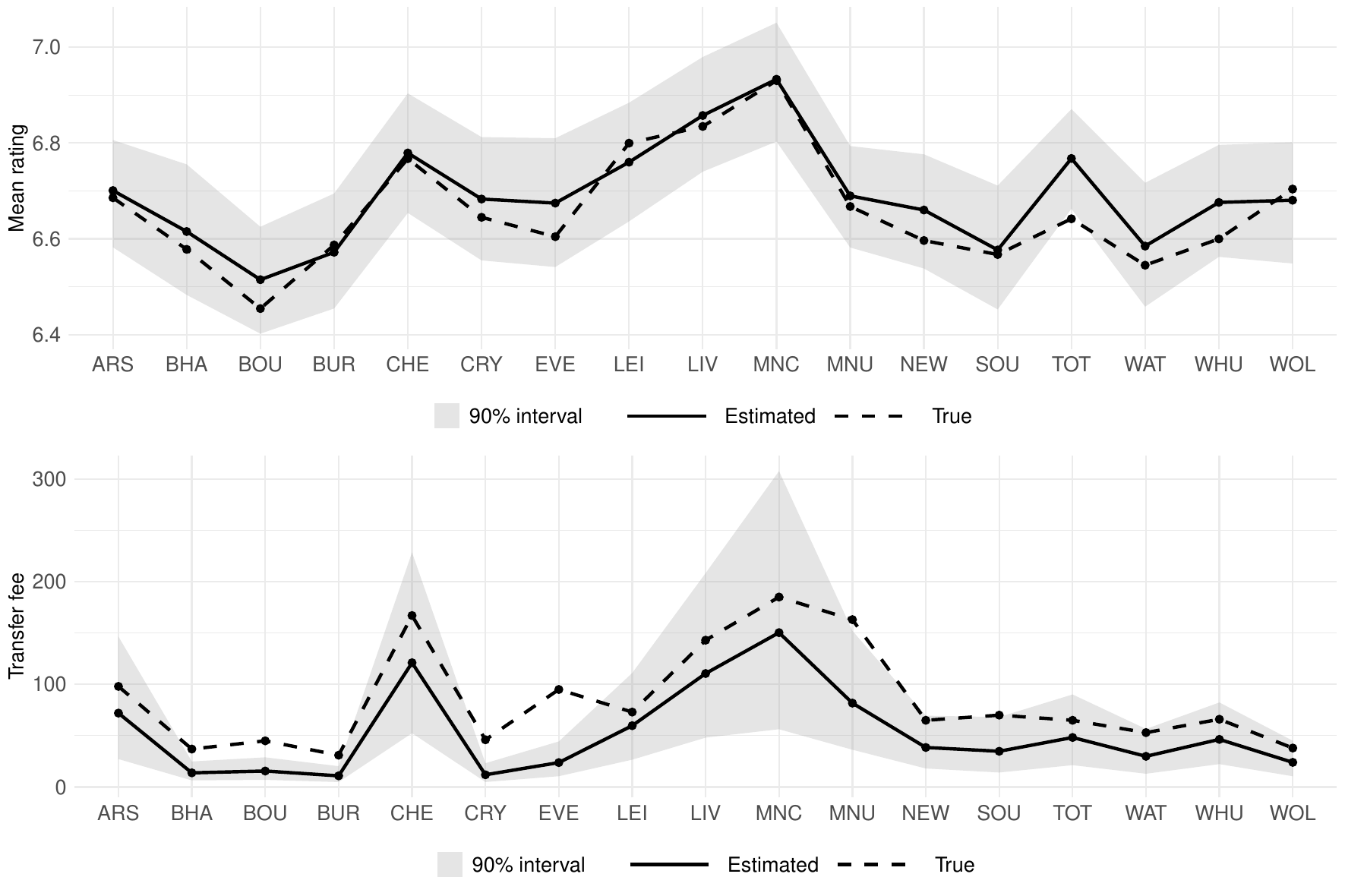}
    \caption{Observed (dashed) versus estimated (solid) club-level means for player ratings (top) and average transfer fees (bottom), with shaded 90\% uncertainty bands around the estimated means. Club codes correspond to Premier League teams in the 2018/19 season.}    
    \label{fig:est_vs_true}
\end{figure}

To begin with, we present an overall comparative summary of our analyses. The aggregated results for the 17 clubs are presented in \Cref{fig:est_vs_true}, which contrasts what each club actually achieved with what it could have achieved under the counterfactual squads suggested by our optimisation algorithm, together with the corresponding counterfactual cost profiles. Here, the clubs are represented by the corresponding three-letter codes. The top panel shows the actual performance of the club in 2019/20 against the estimated performance for the counterfactual squad, along with 90\% prediction interval. In the same fashion, the bottom panel provides a comparison of the actual spending with the estimated spending. Broadly, it can be used as a diagnostic of transfer efficiency: clubs whose realised outcomes in both performance and spending are close to the counterfactual outcomes can be viewed as having conducted efficient business, even if their realised transfers differ from the recommended ones, because they achieved comparable squad strength without requiring materially higher expenditures than those implied by the model-based counterfactual.

A first implication from the figure is that most clubs appear to have had limited headroom for improvement in average on-field performance, in the sense that the counterfactual mean ratings are close to the observed club means. Indeed, the actual performance remains within the confidence band of the predicted performance for all clubs, with an observable difference between the predicted mean performance and the actual for only five clubs -- Bournemouth (BOU), Everton (EVE), Newcastle United (NEW), Tottenham Hotspur (TOT), and West Ham United (WHU). Overall, the top panel suggests that the realised squad construction has been near the predicted frontier defined by our rating model and constraints, so the optimisation algorithm cannot extract large additional gains without relaxing constraints or expanding the feasible player pool. Another possible argument \citep[in line with][]{wanat2022short} is that the clubs, on average, perform at a predictable level and transfer strategies may not materially improve their overall performance. Thus, from a planning perspective, a club may view our framework primarily as a robustness check and a tool for identifying marginal improvements or risk-reduction opportunities rather than as a source of large performance jumps. 

A second implication, based on the bottom panel of \Cref{fig:est_vs_true}, concerns spending discipline. The counterfactual cost estimates may be interpreted as the model-implied expenditure required to implement a squad that achieves the corresponding counterfactual rating. When a club's realised average cost substantially exceeds the counterfactual estimated cost while delivering similar realised performance, this indicates potential overpayment relative to the model's benchmark -- that is, the club may have paid a premium not justified by commensurate squad-level performance. Conversely, when realised spending is at or below the model-implied counterfactual cost for comparable performance, the club can be viewed as having extracted value in the market, either through superior scouting, negotiation, timing, or by exploiting segments where the fee model predicts lower prices for similar projected contributions. Importantly, this inference remains valid even when the realised transfer targets differ from the recommended set, because the comparison is conducted at the level of squad aggregates and therefore evaluates outcomes rather than matching individual transactions. Further on this note, the uncertainty bands are informative for identifying clubs whose realised spending is difficult to rationalise under the counterfactual benchmark. For instance, top clubs like Manchester United and Everton exhibit realised average costs that sit well above the upper end of the counterfactual bands implied by our fee model for the recommended squads, while the corresponding counterfactual mean ratings are not dramatically different from the observed squad averages. Interpreted through our framework, this pattern is consistent with potential overpayment: the clubs appear to have paid a substantial premium relative to what would have been required -- under the model-implied market conditions and the optimisation recommendations -- to attain similar squad-level performance. Clubs such as Brighton, Bournemouth, Burnley, and Crystal Palace also show significantly higher fees paid in transfer, which may be suggestive of systematic inefficiency. For the other clubs, even if the realised transfers differ from our recommended set, the clubs can be viewed as having achieved comparable squad outcomes without spending materially above the benchmark implied by our models. Particularly for Chelsea, Liverpool, and Manchester City, it is critical to note that the predicted fees have great uncertainty (possibly due to targeting high-value players), but the teams managed to maintain their expenditure at a manageable level.

While \Cref{fig:est_vs_true} gives us an overall idea of the performance of the optimisation routine in identifying a balanced squad with budget constraints, we next look at the decision-making aspect of the optimisation routine involving the transfer of individual players and these results are illustrated in \Cref{tab:transfer_comparison}. For brevity of space, we include the list of players who are estimated to be over $\text{\euro} 20$ million by the transfer fee model. The results are separated into three categories to showcase the alignment and the deviation between optimisation recommendations and recorded transfers of the players.

\begin{table}[!ht]
\centering
\caption{Comparison of recommended (according to our approach, after the 2018/19 season) and recorded transfers for players valued above $\text{\euro} 20$ million. The year in the parentheses in Panel B indicates when the player was eventually transferred, `Free' indicates the player left as a free agent, `Loan' indicates the player was loaned out, and a blank suggests the player is still playing for that club.}
\label{tab:transfer_comparison}
{\scriptsize
\begin{tabular}{llccccc}
\toprule
Team & Player & Position & Expected Price & IQR & Actual Price & Rating \\
\midrule
\multicolumn{7}{c}{\textit{Panel A: Recommended and sold}} \\
\midrule
Chelsea & Eden Hazard & Forward & 208 & [58, 243] & 121 & 7.68 \\
Crystal Palace & Aaron Wan-Bissaka & Defender & 24 & [7, 28] & 55 & 7.52 \\
Leicester City & Harry Maguire & Defender & 54 & [15, 63] & 87 & 7.08 \\
Manchester City & Fabian Delph & Midfielder & 29 & [8, 34] & 10 & 7.16 \\
Manchester United & Romelu Lukaku & Forward & 80 & [22, 93] & 74 & 7.29 \\
Manchester United & Chris Smalling & Defender & 31 & [9, 36] & Loan & 7.24 \\
Southampton & Mario Lemina & Midfielder & 32 & [9, 38] & Loan & 6.93 \\
Tottenham Hotspur & Kieran Trippier & Defender & 44 & [12, 51] & 22 & 7.04 \\
West Ham United & Andy Carroll & Forward & 26 & [7, 30] & Free & 6.69 \\
\midrule
\multicolumn{7}{c}{\textit{Panel B: Recommended but not sold}} \\
\midrule
Liverpool & Sadio Mane & Forward & 76 & [21, 89] & 32 (2022) & 7.27 \\
Liverpool & Andrew Robertson & Defender & 28 & [8, 32] & -- & 7.02 \\
Liverpool & Jordan Henderson & Midfielder & 23 & [7, 27] & 14 (2023) & 6.95 \\
Manchester City & Bernardo Silva & Midfielder & 59 & [16, 69] & -- & 6.79 \\
Manchester City & Riyad Mahrez & Forward & 53 & [15, 61] & 35 (2023) & 7.28 \\
Southampton & Pierre Emile Hojbjerg & Midfielder & 27 & [7, 31] & 17 (2020) & 6.61 \\
\midrule
\multicolumn{7}{c}{\textit{Panel C: Not recommended but sold}} \\
\midrule
Arsenal & Alex Iwobi & Midfielder & 24 & [7, 28] & 30 & 6.88 \\
Arsenal & Aaron Ramsey & Midfielder & 20 & [6, 24] & Free & 7.27 \\
Chelsea & David Luiz & Defender & 43 & [12, 50] & 9 & 6.85 \\
Liverpool & Daniel Sturridge & Forward & 30 & [8, 34] & Free & 6.76 \\
Manchester City & Danilo & Defender & 74 & [21, 86] & 37 & 6.91 \\
Manchester United & Alexis Sanchez & Forward & 49 & [14, 58] & Loan & 7.55 \\
Manchester United & Ander Herrera & Midfielder & 35 & [10, 40] & Free & 6.60 \\
Manchester United & Marouane Fellaini & Midfielder & 28 & [8, 33] & 7 & 6.94 \\
Newcastle United & Ayoze Perez & Forward & 26 & [7, 30] & 33 & 6.74 \\
Southampton & Charlie Austin & Forward & 25 & [7, 29] & 4 & 6.53 \\
Southampton & Matt Targett & Defender & 24 & [7, 28] & 16 & 6.96 \\
Tottenham Hotspur & Mousa Dembele & Midfielder & 28 & [8, 33] & 5 & 6.90 \\
West Ham United & Marko Arnautovic & Forward & 32 & [9, 37] & 25 & 7.07 \\
\bottomrule
\end{tabular}
}
\end{table}

In Panel~A, we present the cases where the recommended sales of the players coincided with the realised transfers. These instances indicate that several observed transfer movements for the EPL clubs align with the recommendations from our optimisation routine, thereby highlighting the practical relevance of the constraints incorporated within our optimisation framework. Although there is a considerable difference between the estimated and recorded transfer fees for these players, most recorded transfer prices remain within the corresponding interquartile ranges predicted by the model, suggesting that the transfer-fee framework provides a reasonable approximation of market uncertainty. We however notice an anomalous pattern in the transfers of Harry Maguire and Aaron Wan-Bissaka. Both players were bought by Manchester United for amounts exceeding the upper end of the interquartile range, which may indicate weak negotiation behaviour or overpayment by Manchester United during that transfer window \citep[][reported similar findings for Maguire]{McHaleHolmes2023}. This aligns with our earlier observation as well, wherein \Cref{fig:est_vs_true} we noticed that Manchester United's actual spending was more than the upper confidence limit recommended by our methodology. Interestingly, since both players are defenders, these transfers may indicate that Manchester United strategically prioritised strengthening defensive depth during that season, while operating under relatively flexible financial constraints. On the other hand, it is also interesting that both players are Englishmen by birth -- \cite{mittal2022high} calls similar spending patterns an ``English tax'' which says that the transfer of English players within the EPL is more inflated than other nations. 

As we move on to Panel~B, it shows the players who are recommended to be sold by our algorithm but were retained during that season by their respective clubs. We however note that four of the six players were eventually sold -- one in the immediate next season and the others a few years later. These cases highlight situations where decision-makers may prioritise non-financial objectives, including squad stability or competitive considerations, over immediate economic returns. The recorded delays in eventual transfers for some players, albeit with lower recorded transfer fees than expected fees, further suggest that transfer decisions may often feel unsuitable at times, but declining performance of the players can lead to poor business in the next transfer windows. We also prefer to discuss these results with Panel~C, which includes players who were transferred after the 2018/19 season, but were not recommended for sale by our method. Three of these players left as free agents, indicating that the lack of contract information affected the output of our methods. Specifically, for Liverpool, Daniel Sturridge was preferred by the model to be kept in the squad over Sadio Mane, but since Sturridge's contract was not renewed, he left the club, and Mane possibly became more important considering their similar positions in the lineup, thereby not initiating his transfer out of the club until 2022. Thus, such deviations from model-optimal decisions are potentially driven by external constraints such as wage commitments, contract expiration, or exogenous strategic adjustments. For example, a new management team may no longer prefer certain types of players in their game plans and may choose to transfer out such players as soon as possible to avoid paying their wages.

In general, the results illustrate that the framework provides a useful public data decision benchmark for transfer planning. Recorded transfers reflect a combination of economic optimisation, private institutional constraints, contractual considerations, and dynamic strategic adjustments that are not fully captured by publicly observable variables. The discrepancies between model recommendations and realised transfers are therefore not evidence of model failure; rather, they delineate the boundary of a public-data decision-support model and highlight where club-specific private information -- wages, contract state, tactical plans, player preferences -- would meaningfully refine the recommendations.

\subsection{Competition between clubs}

As the final piece of discussion in this application, to illustrate and understand the effects of multiple clubs interested in buying the same player, we simulate a single-round independent private value auction for two players: Miguel Almiron and Adama Traore. Note that both players are recommended to buy for multiple EPL clubs in our dataset. A complete list of all such players and related discussions is provided in Section S4 of the supplement. In this case, we follow a simulation framework wherein for each player we draw $N=2000$ independent valuations from the estimated log-normal distributions (below denoted as $\mathcal{LN}$). We compute equilibrium bids by inverting the boundary value problem as described in the supplement (provided in Section S5, as an extension to the discussions in \Cref{sec:bidding}), and present bidding strategy, sale outcomes and prices. 

For Almiron, Newcastle United (NEW) expects a transfer fee of $\text{\euro} 8.2$ million, with the reserve price distribution given by $\rho \sim \mathcal{LN}\left(1.5, 1.1^2\right)$. Three clubs have been recommended to buy him: Everton (EVE), with valuation $\mathcal{LN}\left(1.2, 1.1^2\right)$; Southampton (SOU), with valuation $\mathcal{LN}\left(1.7, 1.1^2\right)$; and Watford (WAT), with valuation $\mathcal{LN}\left(1.5, 1.1^2\right)$. Newcastle's affinity toward these clubs is modelled using logistic distributions with scale parameter $1$. The centres of the logistic distributions for EVE, SOU, and WAT are $2.7$, $4.3$, and $3.4$, respectively. In other words, while SOU is the strongest bidder, NEW exhibits the least affinity towards them. Alternatively, the situation can be viewed as a purely financial decision. Since SOU is financially stronger than the other clubs, NEW expects to secure a higher transfer fee from SOU than from the smaller clubs. Similarly, for Traore, his present club Wolverhampton Wanderers (WOL) expects a transfer fee of $\text{\euro} 3.7$ million, with the reserve price distribution given by $\rho \sim \mathcal{LN}\left(0.7, 1.1^2\right)$. Four clubs are recommended to buy him: Manchester City (MNC), with valuation $\mathcal{LN}\left(1.8, 1.1^2\right)$; Tottenham Hotspur (TOT), with valuation $\mathcal{LN}\left(1.3, 1.1^2\right)$; NEW with valuation $\mathcal{LN}\left(1.2, 1.1^2\right)$; and WAT with valuation $\mathcal{LN}\left(1.2, 1.1^2\right)$. The seller's affinity values toward these four clubs are centred at $5.1$, $2.9$, $2.6$, and $2.7$, respectively. Therefore, we may say that MNC is the strongest bidder and WOL has the least affinity towards them. 

\begin{figure}[!h]
    \centering
    \begin{subfigure}{\linewidth}
        \centering
        \includegraphics[width=0.8\textwidth]{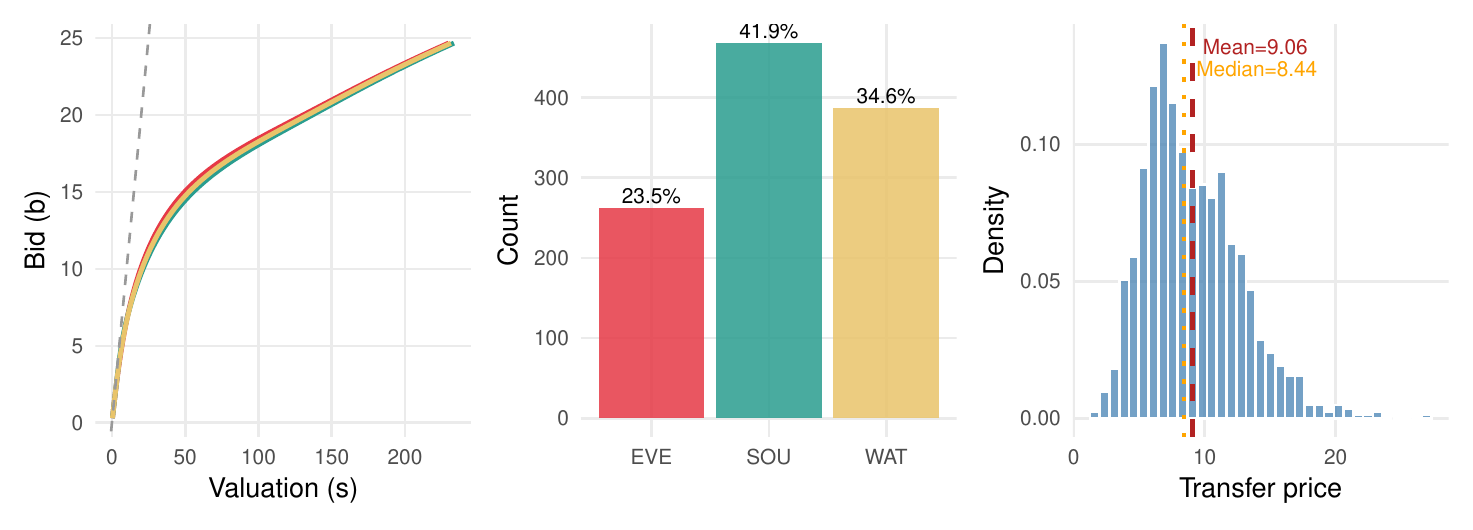}
        \caption{Analysis of auction mechanism for Almiron}
        \label{fig:bid:almiron}
    \end{subfigure}
    \vspace{0.5cm}
    \begin{subfigure}{\linewidth}
        \centering
        \includegraphics[width=0.8\textwidth]{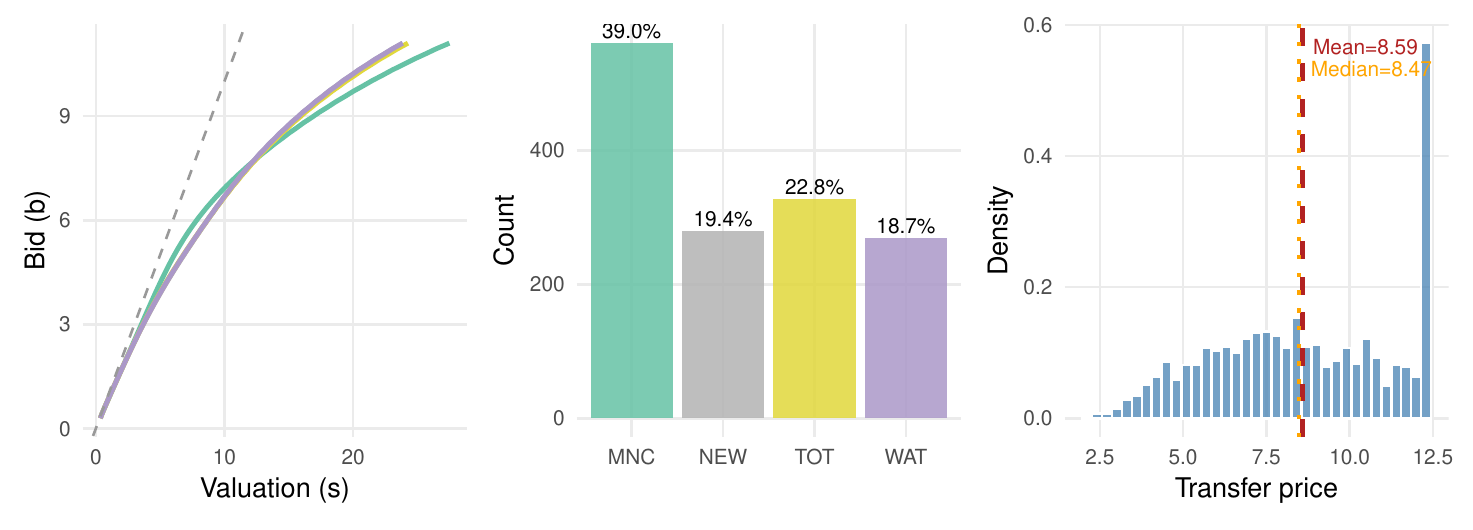}
        \caption{Analysis of the auction mechanism for Traore}
        \label{fig:bid:traore}
    \end{subfigure}
    \caption{Summary of simulation studies for understanding transfer market dynamics. The left panel illustrates the bidding strategy of different clubs; the middle panel illustrates the win probabilities of different clubs; and the right panel illustrates the histograms of transfer fees.}
    \label{fig:bid:combined}
\end{figure}

We present our analyses in \Cref{fig:bid:combined} with three separate panels. In the left-most panel, we present the bidding strategies of different clubs. We notice that, in line with \Cref{rem:maskin:riley}, the Maskin-Riley ordering holds for our analysis with Almiron. However, for Traore, as WOL is less affine towards MNC, the latter are required to bid more aggressively for lower values of bids to increase the chance of winning. MNC can afford to bid more conservatively for higher prices as they are the most dominant bidder in terms of the valuation distribution. In the middle panel, we present the conditional sell probability to each club. Overall, Traore is sold in 71.8\% of the simulated scenarios, as compared to 55.8\% for Almiron. This substantial gap can be attributed to the reserve distributions: while Almiron's seller has $\mu_\rho=1.5$ (expected reserve $\text{\euro} 8.2$ million), Traore's seller has $\mu_\rho=0.7$ (expected reserve $\text{\euro} 3.7$ million). A weaker reserve distribution is dominated by the equilibrium bids more frequently. Within the winning transfers, SOU is the most dominant club to acquire Almiron with almost 42\% winning rate, whereas MNC is the most dominant club to acquire Traore with almost 39\% winning rate. Finally, in the rightmost panel, we present the price distributions. We see that Almiron generates a higher mean price ($\text{\euro} 9.1$ million versus $\text{\euro} 8.6$ million for Traore) despite a lower sale probability. As discussed earlier, the mechanism is right-truncated by an amount mimicking the buyout clause. This value is lower for Traore ($\text{\euro} 12.3$ million), which means many bids clear this amount. As a result, we see a sharp peak for him at this point. For the convenience of readers, we also summarise these results in \Cref{tab:summary:bids:combo}.

\begin{table}[h]
\centering
\caption{Summary of single-round simulation results for two players with multiple interests.}
\label{tab:summary:bids:combo}
{\scriptsize
\begin{tabular}{lcc}
\toprule
& Almiron & Traore \\
\midrule
Bidding clubs        & 3      & 4      \\
Sale probability           & 55.8\% & 71.8\% \\
Expected price ({\euro} million)    &  8.2 &  3.7 \\
Mean selling price ({\euro} million)    &  9.1 &  8.6 \\
Standard deviation of price ({\euro} million)      &  3.7 &  2.7 \\
IQR of price ({\euro} million) &  [6.3, 11.4]  & [6.5, 11]  \\
Dominant club (share)      & SOU (41.9\%) & MNC (39.0\%) \\
\bottomrule
\end{tabular}
}
\end{table}

Both simulation case studies provide useful insights into the mechanisms through which football transfers may generate inflated transfer fees. In the case of Almiron, the simulated outcomes remain closely aligned with the seller club's reserve-price distribution and the final estimated transfer fee, indicating a relatively balanced negotiation environment. In contrast, the analysis of Traore demonstrates that the presence of a highly dominant club in the bidding process can substantially inflate the final transfer fee. In our simulated setting, the resulting transfer price increases to more than twice the level implied by the seller club’s reserve distribution. This highlights how asymmetries in financial strength and valuation across clubs can significantly distort market prices during a transfer window. Such results can further be extended to a multi-round negotiation paradigm, which we briefly discuss in the supplementary material (refer to Section S5).

\section{Concluding remarks}\label{sec:conclusions}

In this paper, we have proposed a unified quantitative framework for football transfer strategy that combines player-performance modelling, transfer-fee estimation, constrained squad optimisation, and competitive bidding dynamics within a single decision-analytic structure. Using mixed-effects models for both player ratings and transfer fees, we incorporated player-level characteristics, recent performance indicators, club and league environments, and latent heterogeneity across transfer corridors. These predictive components were subsequently integrated into a weighted multi-criteria constrained optimisation framework that balances squad quality, transfer expenditure, and financial risk under realistic operational constraints. Finally, motivated by the competitive nature of modern football transfers, we embedded the outcomes of the optimisation framework within an asymmetric independent private-value auction model with random reserve prices to study transfer-market behaviour when multiple clubs compete for the same player.

Our empirical analysis shows several important findings. First, the rating model confirms that player contribution depends jointly on individual trajectory and squad context, thereby supporting the portfolio-based interpretation of squad construction. Second, the transfer-fee model highlights the importance of persistent quality signals, market inflation, league-specific pricing structures, and seller bargaining strength in determining transfer fees. Third, the optimisation framework produces counterfactual squad configurations that often achieve comparable or improved projected squad quality while requiring lower expenditure than recorded transfer activity. This suggests that several clubs may spend inefficiently relative to a model-implied benchmark, even when their sporting outcomes remain competitive. The auction-theoretic extension further illustrates how transfer-market dynamics can generate substantial variation in transfer fees. In particular, the simulations demonstrate that the participation of financially dominant clubs can significantly alter equilibrium bidding behaviour and increase transfer prices beyond baseline reserve-price expectations. At the same time, the introduction of affinity-dependent acceptance probabilities provides a flexible mechanism for capturing strategic seller preferences, including reluctance to sell to direct rivals or willingness to negotiate more with specific clubs. Consequently, the framework offers a probabilistic interpretation of transfer negotiations that more closely resembles the institutional structure of football transfers than classical auction formulations with deterministic reserve prices.

Nevertheless, several limitations remain. The current optimisation framework does not explicitly incorporate wages, contract duration, agent fees, injury histories, tactical compatibility, or dynamic interactions across negotiations for multiple players by the same club. In addition, the choice of mixing parameters is heuristic and intended to reflect stylised differences in financial behaviour across clubs. Future work could address these limitations through dynamic multi-period optimisation, data-driven calibration of objective weights, Bayesian uncertainty propagation, tactical interaction networks, and sequential or multi-round models of transfer negotiations involving multiple players.

Overall, the proposed methodology demonstrates how statistical modelling, stochastic optimisation, and auction theory can be integrated to study football transfers within a unified operational framework. As transfer markets continue to become more data-driven and financially competitive, quantitative decision-support systems of this nature can help clubs evaluate transfer strategies while balancing sporting ambition with long-term financial sustainability.

\section*{Data availability statement}

The data used in this research have been obtained from two sources: Transfermarkt website (link: \url{https://www.transfermarkt.co.uk/}) and WhoScored website (link: \url{https://www.whoscored.com/}). The cleaned data, along with appropriate R scripts, will be made publicly available in a GitHub repository maintained by the first author.

\bibliography{references}
\bibliographystyle{apalike}

\newpage

\begin{center}
	{\LARGE{\bf Supplementary materials}}
\end{center}

\setcounter{table}{0}
\renewcommand{\thetable}{S.\arabic{table}}
\setcounter{figure}{0}
\renewcommand{\thefigure}{S.\arabic{figure}}
\setcounter{section}{0}
\renewcommand{\thesection}{S\arabic{section}}
\setcounter{equation}{0}
\renewcommand{\theequation}{S.\arabic{equation}}

\setcounter{proposition}{0}

\section{Theoretical results and their proofs}

In Section 2.4 of the main paper, we discuss two key results that give us a framework for solving the equilibrium strategy in competitive bidding using numerical integration techniques. Those two results and their proofs are provided below in detail.

\begin{proposition}
\label{prop:affinity-foc}
Suppose $F$ and $H$ are continuously differentiable with strictly positive densities and $p_c(\cdot)$ and $\psi_c(\cdot)$ are continuously differentiable function on $(0, \upsilon_{\mathsf{thresh}}]$. Then in an interior pure-strategy equilibrium under risk neutrality, the inverse bidding functions satisfy
\begin{equation}
\frac{1}{s-b}
=
\frac{p_c'(b)}{p_c(b)}
+
\frac{h(b)}{H(b)}
+
\sum_{j\neq c}\frac{f_j(\psi_j(b))}{F_j(\psi_j(b))}\psi'_j(b) .
\label{eq:foc}
\end{equation}
Moreover, if $p_c(b)\equiv p_c$ is constant, equilibrium bids coincide
with those obtained under $p_c(b)\equiv1$; that is only expected utilities are
scaled by $p_c$.
\end{proposition}

\begin{proof}
Differentiating $U_c(b;s)$ given by equation (21) in Section~2.4 of the main manuscript yields
\begin{equation*}
    \frac{d}{db}U_c(b;s)=-q_c(b)+(s-b)q_c'(b).
\end{equation*}
At the optimum, $q_c'(b)/q_c(b)=1/(s-b)$. Then, using the identity in equation (20) of the main paper and logarithmic differentiation, we can write
\begin{align*}
    \frac{q_c'(b)}{q_c(b)}
&=
\frac{d}{db}\log \hat{p}_c(b)
+
\frac{d}{db}\log \hat{H}(b)
+
\sum_{j\neq c}\frac{d}{db}\log\left(
F_j(\psi_j(b))\right)\\
&=
\frac{p_c'(b)}{p_c(b)}
+
\frac{h(b)}{H(b)}
+
\sum_{j\neq c}\frac{f_j(\psi_j(b))}{F_j(\psi_j(b))}\psi'_j(b),
\end{align*}
which proves \eqref{eq:foc}.
If $p_k(b)$ is constant, it factors out of the objective and does not
affect the maximiser.
\end{proof}

\begin{proposition}
\label{thm:asymm-ode}
Suppose $F$ and $H$ are continuously differentiable with strictly
positive densities on
$(\underline{s},\overline{s}]$,
and $p(\cdot)$ and $\psi_c(\cdot)$ are continuously differentiable in the interval $[b_{\min}, \upsilon_{\mathsf{thresh}}]$ with $p(b), \psi_c(b)>0$. Then, we get the following system of ordinary differential equations for $c\in\mathcal{C}^i$:
\begin{equation*}
    \psi'_c(b) = \frac{1}{\left(\abs{\mathcal{C}^i} - 1\right)}\frac{F_c(\psi_c(b))}{f_c(\psi_c(b))}\left[\sum_{j\not = c}\left(\frac{1}{\psi_j(b)-b} - \frac{p'_j(b)}{p_j(b)}\right)-
    \left(\abs{\mathcal{C}^i}-2\right)\left(\frac{1}{\psi_c(b)-b} - \frac{p'_c(b)}{p_c(b)} \right) -
    \frac{h(b)}{H(b)}
    \right].
\end{equation*}
\end{proposition}

\begin{proof}
At Bayes Nash equilibrium, we have $\psi_c(b) = s$. Then, we can write \eqref{eq:foc} as
\begin{equation}\label{eq:equib:one}
\frac{1}{\psi_c(b)-b}
=
\frac{p_c'(b)}{p_c(b)}
+
\frac{h(b)}{H(b)}
+
\sum_{j\neq c}\frac{f_j(\psi_j(b))}{F_j(\psi_j(b))}\psi'_j(b).
\end{equation}
Summing over all values of $c\in\mathcal{C}^i$, we get:
\begin{equation*}
    \sum_{j}\frac{1}{\psi_j(b)-b} = \sum_j \frac{p'_j(b)}{p_j(b)} + \abs{\mathcal{C}^i}\frac{h(b)}{H(b)} 
    + \left(\abs{\mathcal{C}^i} - 1\right)\sum_j \frac{f_j(\psi_j(b))}{F_j(\psi_j(b))}\psi'_j(b),
\end{equation*}
or, equivalently,
\begin{equation}\label{eq:equib:all}
    \frac{1}{\left(\abs{\mathcal{C}^i} - 1\right)}\sum_{j}\frac{1}{\psi_j(b)-b}
    = \frac{1}{\left(\abs{\mathcal{C}^i} - 1\right)}\sum_j \frac{p'_j(b)}{p_j(b)} +
    \frac{\abs{\mathcal{C}^i}}{\left(\abs{\mathcal{C}^i} - 1\right)}\frac{h(b)}{H(b)}  +
    \sum_j \frac{f_j(\psi_j(b))}{F_j(\psi_j(b))}\psi'_j(b).
\end{equation}
Then, subtracting \eqref{eq:equib:one} from \eqref{eq:equib:all}, we get
\begin{equation*}
    \frac{1}{\left(\abs{\mathcal{C}^i} - 1\right)}\left[\sum_{j}\frac{1}{\psi_j(b)-b} - \sum_j \frac{p'_j(b)}{p_j(b)}\right] - 
    \left[\frac{1}{\psi_c(b)-b} - \frac{p'_c(b)}{p_c(b)}\right] -
    \frac{1}{\left(\abs{\mathcal{C}^i} - 1\right)}\frac{h(b)}{H(b)} = \frac{f_c(\psi_c(b))}{F_c(\psi_c(b))}\psi'_c(b),
\end{equation*}
which leads to the following equation to complete the proof:
\begin{equation*}
    \psi'_c(b) = \frac{1}{\left(\abs{\mathcal{C}^i} - 1\right)}\frac{F_c(\psi_c(b))}{f_c(\psi_c(b))}\left[\sum_{j\not = c}\left(\frac{1}{\psi_j(b)-b} - \frac{p'_j(b)}{p_j(b)}\right)-
    \left(\abs{\mathcal{C}^i}-2\right)\left(\frac{1}{\psi_c(b)-b} - \frac{p'_c(b)}{p_c(b)} \right) -
    \frac{h(b)}{H(b)}
    \right].
\end{equation*}
\end{proof}

\section{Additional discussions on the model for player ratings}

In Section 3.1 of the main paper, we discussed the key results of the linear mixed effects model used for analysing the player ratings. Here, we take a look at the estimated variance components, which are helpful to obtain a concise summary of the magnitude of latent heterogeneity captured by the random intercepts. 

Overall, the residual standard deviation is $\hat{\sigma}=0.2996$, which dominates the decomposition and indicates that most unexplained variation remains at the player-season level after accounting for observed covariates and random intercepts. The league-level random intercepts are non-negligible: the current-league effect has standard deviation $0.0929$ while the last-league effect has standard deviation $0.0846$. Relative to the total variance $\hat{\sigma}^2+\hat{\sigma}^2_{\mathrm{cur}}+\hat{\sigma}^2_{\mathrm{last}}+\hat{\sigma}^2_{\mathrm{club}}$, these correspond to variance shares of approximately $8.2\%$ (current league) and $6.8\%$ (last league), with the residual accounting for about $84.8\%$. In contrast, the origin-destination club-corridor random intercept has a standard deviation of $0.0122$, contributing only about $0.14\%$ of the total variance. This indicates that, conditional on the fixed effects (including lagged rating and team context) and league environments, there is little additional persistent corridor-specific intercept shift left to explain. 

From the perspective of our transfer strategy pipeline, the league random effects therefore play a substantive role in producing coherent counterfactual forecasts across league environments, while the corridor term mainly serves as a robustness layer that can adjust a small subset of historically atypical corridors without materially affecting most predictions. The empirical distributions of the best linear unbiased predictors (BLUPs) are presented in \Cref{fig:blup_rating}. This is also consistent with this decomposition: the league BLUPs exhibit visibly wider dispersion around zero than the club-corridor BLUPs, which are tightly concentrated near the origin.

\begin{figure}[!ht]
    \centering
    \includegraphics[width=0.9\textwidth]{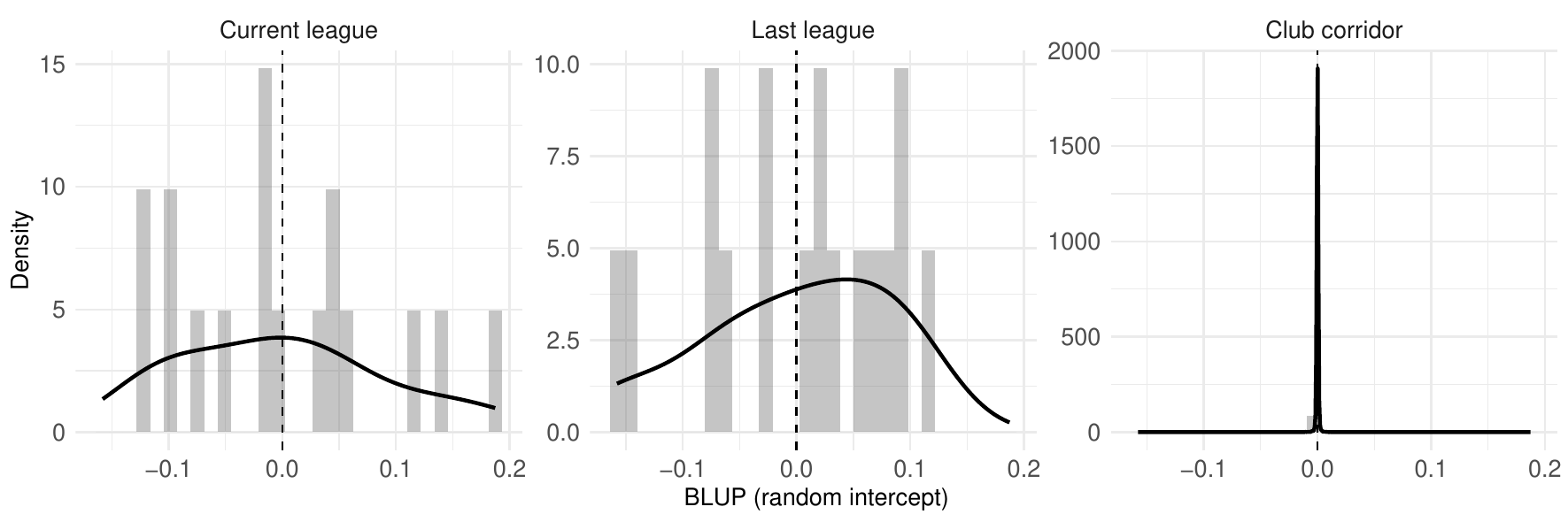}
    \caption{For the rating model, empirical densities of the BLUPs (random intercept estimates for current league, last league, and origin-destination club corridor) are presented left to right.}
    \label{fig:blup_rating}
\end{figure}

Next, we find it imperative to benchmark the \proposed\ rating model against other commonly used alternatives. The first comparator is a log-linear ordinary least squares regression model (hereafter denoted as \OLS) that includes the same fixed effects as the \proposed\ model but does not incorporate any random effects, and therefore treats the rating scale as homogeneous across league environments. The second comparator \lasso\ is a purely penalised regression model that performs variable selection and shrinkage through an $\mathcal{L}_1$ penalty, serving as a representative high-dimensional baseline. The third comparator \enet\ improves upon \lasso, where we take a two-stage procedure in which the fixed-effect feature set is first reduced using the elastic-net regularisation and the resulting predictors are then used in a mixed-effects model, aiming to trade interpretability for parsimony. We note that these are common methods used for predicting future performances of football players \citep[see, e.g.,][]{bhatnagar2024rating, klaiber2026quantifying}. We also examined other techniques such as Ridge regression, random forest, XGBoost, and beta regression, but their efficacy was markedly worse than the proposed model, and the results are omitted from here for brevity. The following comparison is intended as a predictive plausibility check rather than a claim of universal superiority. All methods are evaluated on the same held-out test set, and differences in performance across rating strata should be interpreted in light of the downstream optimisation context, where ranking accuracy among high-rated players is most consequential for squad-selection decisions.

The above-mentioned four methods are compared in terms of fit and predictive accuracy. The first aspect is judged based on the root mean squared error (RMSE) of the fitted values on the training set, and the Akaike information criterion \citep{stoica2004model}. The second aspect is assessed based on the RMSE of the predicted values with respect to the truth, and the mean absolute error (MAE) of the same for the entire test set. To further understand the accuracy in different rating groups, we also compute the MAE for four different brackets, as shown in \Cref{tab:model-comparison-rating}.

\begin{table}[!ht]
    \centering
    \caption{Comparison between different models for rating prediction in training and test set.}
    \label{tab:model-comparison-rating}
    {\scriptsize
    \begin{tabular}{llccccc}
    \toprule
    Evaluation set & Metric & \proposed\ & \OLS\ & \lasso\ & \enet\ \\
    \midrule
    Training set & RMSE & 0.299 & 0.299 & 0.380 & 0.344 \\
     & AIC & 16864.56 & 16912.20 & 52101.91 & 31040.34 \\
     \midrule
    Test set & MAE & 0.229 & 0.230 & 0.286 & 0.283 \\
     & RMSE & 0.296 & 0.298 & 0.364 & 0.361 \\
    \midrule
    Test set with rating $<$ 6 & MAE & 0.584 & 0.584 & 0.728 & 0.766 \\
    Test set with rating 6--6.5 & MAE & 0.243 & 0.242 & 0.310 & 0.354 \\
    Test set with rating  6.5--7 & MAE & 0.160 & 0.162 & 0.157 & 0.135 \\
    Test set with rating  $>$ 7 & MAE & 0.402 & 0.402 & 0.615 & 0.532 \\
    \bottomrule
    \end{tabular}}
\end{table}

The results demonstrate that the \proposed\ model provides the best overall balance between goodness-of-fit and out-of-sample accuracy. On the training set, it records the best in-sample RMSE while achieving the smallest AIC among the likelihood-based candidates, indicating that the additional random-effect structure improves model fit without overfitting relative to \OLS. In the test set, overall the lowest RMSE and MAE are recorded by the \proposed\ model. The results are similar for \OLS\ as well. These conclusions remain competitive across rating strata. In contrast, the regularisation-based alternatives exhibit noticeably larger test errors, suggesting that aggressive shrinkage and/or misalignment between selection and the hierarchical structure leads to loss of predictive accuracy in this application. The stratified MAE results are particularly relevant for downstream squad optimisation, where ranking errors among high or low-rated players can materially affect selection decisions under budget constraints. Our \proposed\ model maintains stable performance across the rating groups, including the tails, whereas the penalised baselines display larger errors for the extreme segments. Consequently, our approach provides more reliable counterfactual forecasts of player value for the focal club, which strengthens the quality input to the optimisation module and reduces the risk that transfer recommendations are driven by modelling artefacts rather than substantive differences in expected contribution.

\section{Additional discussions on the model for transfer valuation}

Akin to the previous section, we now move on to the discussions on the random effects of the linear mixed effects model used to analyse the transfer valuation. Note that the main insights have been added in Section 3.2 of the main manuscript. 

In this case, there are two random effects, corresponding to the buyer club and the seller club. These random intercepts capture persistent heterogeneity in willingness-to-pay and selling premiums that are not explained by observed player characteristics, recent performance, or the league-level price controls. The estimated standard deviations are $0.4498$ for the buyer-club effect and $0.2983$ for the seller-club effect, while the residual standard deviation on the log-fee scale is $1.0608$. If we translate these numbers to the calculation of total variance, we find that buyer and seller club effects account for approximately $14.3\%$ and $6.3\%$, respectively, with the remaining $79.4\%$ attributable to transaction-level idiosyncratic variation. Interpreting magnitudes on the original fee scale, a one-standard-deviation increase in the buyer effect corresponds to an approximate multiplicative premium of $56.8\%$ in expected fees (keeping all else equal), while a one-standard-deviation increase in the seller effect corresponds to a premium of $34.7\%$ in the transfer fee. 

The empirical BLUP distributions in \Cref{fig:blup_transfer} align with these estimates: the buyer-club BLUPs exhibit visibly wider dispersion around zero than the seller-club BLUPs, indicating stronger cross-club heterogeneity on the demand side than on the supply side. From the standpoint of transfer planning, these random effects provide a parsimonious correction for club-specific pricing regimes: they allow the model to penalise systematically expensive buyers and systematically strong sellers when forecasting acquisition costs, which directly improves budget feasibility and risk control in the downstream chance-constrained optimisation.

\begin{figure}[!ht]
    \centering
    \includegraphics[width=0.7\textwidth]{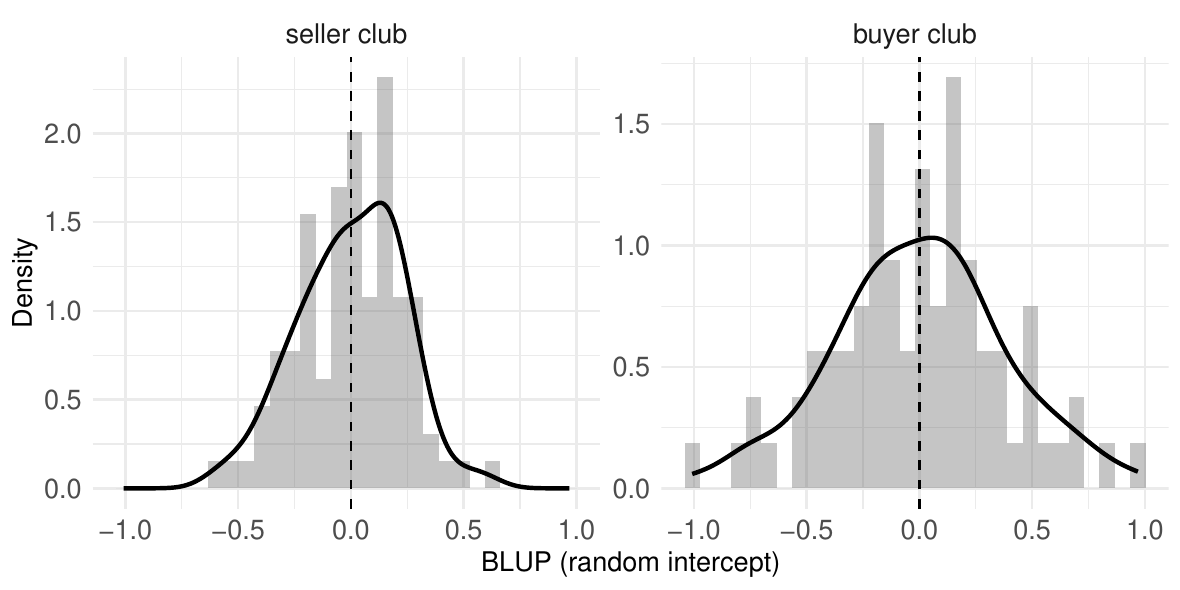}
    \caption{For the transfer-fee model, empirical densities of the BLUPs (random intercept estimates for seller-club and buyer-club effects on the log-fee scale) are presented left to right.}
    \label{fig:blup_transfer}
\end{figure}

It is important to note that the large residual component highlights the importance of maintaining explicit uncertainty handling (as done via the log-normal assumption and the chance constraint), since a substantial fraction of fee variation remains transaction-specific even after accounting for rich covariates and club-level heterogeneity. The shared buyer-club random intercept $b^{\mathrm{buy}}_{c_0}$ induces positive correlation across all transfer fees paid by the same club in a given window. The Marlow approximation used in the deterministic equivalent of the chance constraint (Section~2.3 of the main paper) treats these fees as conditionally independent given the fixed-effect predictions, which leads to a slight underestimate of $\sigma^2_T$ and therefore a marginally less conservative budget constraint. Given that the buyer random effect accounts for approximately $14.3\%$ of total fee variance while the residual transaction-level component dominates ($79.4\%$), we treat this as a reasonable working approximation. A conservative analyst can absorb the underestimation by choosing a slightly tighter $\alpha$ (e.g.\ $\alpha = 0.025$ instead of $0.05$) or by adding a fixed buffer to $B_{\max}$ equal to $z_\alpha \cdot \hat\sigma_{\mathrm{buy}} \cdot \sqrt{n_{\mathrm{buy}}} \cdot \bar{\mu}_Y$, where $n_{\mathrm{buy}}$ is the number of incoming players and $\bar\mu_Y$ is the average predicted fee on the original scale.

As the last piece of discussion here, we conduct a comparative discussion of the \proposed\ fee model against five alternatives that span standard econometric baselines and flexible machine-learning approaches. Some of these methods were considered in a similar comparative discussion by \cite{McHaleHolmes2023}. The first comparator in our study is a log-linear ordinary least squares regression (once again, we use \OLS\ to denote this below) fitted using the same covariate set but without club-level random effects, hence treating buyer and seller heterogeneity as fully explained by observables. Second, we consider the \tobit\ regression framework which has been used by multiple researchers \citep[e.g.,][]{carmichael1999labour, ruijg2015determinants} in similar contexts. The next two methods are two penalised regression models (using the same notations as before, denote them as \lasso\ and \enet) that apply shrinkage and variable selection to improve parsimony, at the expense of explicitly modelling hierarchical market structure. Finally, we use the \xgb\ method that provides a non-linear benchmark trained on the same feature set, representing a flexible prediction-focused alternative without an explicit likelihood-based structure for inference. As with the rating model, the comparison below is a predictive benchmark exercise rather than a claim of universal dominance. We note in particular that the \proposed\ model's advantage is concentrated in the upper fee segment ($\geqslant 50$ million euros), which is the region most consequential for budget feasibility and risk in the downstream chance-constrained optimisation; the modest underperformance in the lowest-fee group is less material for squad planning, since these transfers contribute relatively little to aggregate spending.

Here also, we compare the methods using RMSE, AIC and MAE (for the entire test set as well as for different price groups), and the results are presented in \Cref{tab:model-comparison-transfer}. Note that AIC is not computed for \xgb\ as it is not a likelihood-based model.

\begin{table}[!ht]
    \centering
    \caption{Comparison between different models for transfer valuation in training and test set (RMSE and MAE are presented in millions of euros).}
    \label{tab:model-comparison-transfer}
    {\scriptsize
    \begin{tabular}{llcccccc}
    \toprule
    & & \proposed\  & \OLS\ & \tobit & \lasso\ & \enet\ & \xgb\ \\
    \midrule
    Training set & RMSE & 1.029 & 1.151 & 1.151 & 1.541 & 1.551 & 1.469 \\
     & AIC & 7392.79 & 7607.77 & 7595.24 & 9060.72 & 9051.19 & -- \\
     \midrule
    Test set  & MAE & 7.799 & 8.680 & 8.845 & 11.407 & 11.406 & 11.415 \\
     & RMSE & 14.156 & 16.285 & 16.697 & 22.303 & 22.306 & 22.316 \\
    \midrule
    Test set with fees $<$ 10m & MAE & 2.915 & 2.592 & 2.535 & 2.309 & 2.333 & 2.296 \\
    Test set with fees 10m--20m & MAE & 7.851 & 8.563 & 8.309 & 11.156 & 11.171 & 11.231 \\
    Test set with fees 20m--50m & MAE & 15.505 & 18.082 & 18.909 & 24.650 & 24.651 & 24.781 \\
    Test set with fees $\geqslant$ 50m & MAE & 41.924 & 52.404 & 53.952 & 79.667 & 79.705 & 79.565 \\
    \bottomrule
    \end{tabular}}
\end{table}

We observe that the \proposed\ mixed-effects model achieves the best overall predictive accuracy and in-sample fit among the competing specifications. On the training set, it attains the lowest RMSE and also the smallest AIC among likelihood-based candidates, indicating that incorporating buyer and seller club random intercepts improves fit beyond the fixed-effect baseline without undue complexity. In terms of fit, the \tobit\ model is the second best, closely followed by \OLS, whereas the other approaches are much worse. More importantly, the predictive accuracy of the \proposed\ method outperforms all others when we look at the overall accuracy. Interestingly, the penalised regression models and \xgb\ should be preferred for low-valued players as the accuracy is substantially higher for them in such cases. The gap between the \proposed\ model and the remaining approaches widens markedly as transfer fees increase, with especially large gains for transfers above $\text{\euro}$ 50 million. 

The above pattern is consistent with the model structure. Club-level random effects and market-level covariates help capture persistent willingness-to-pay and selling premium, which are most consequential in the upper tail where a small number of financially dominant buyers and strong sellers drive realised prices. From the perspective of downstream optimisation, this is particularly valuable because large-fee targets dominate budget feasibility and risk. The proposed model's improved performance for expensive transfers therefore reduces the likelihood of systematic budget underestimation and leads to more reliable chance-constrained decisions. At the same time, the modest underperformance in the lowest-fee group is less consequential for squad planning, since these transfers contribute relatively little to aggregate spending and can be accommodated through robustness in the optimisation stage.

\section{Additional results related to transfer strategy optimisation}

\subsection{Benchmarking for optimiser}
\label{sec:benchmark}
 
To assess the robustness of the transfer recommendations to the choice of optimisation algorithm, we compare the genetic algorithm (GA) \citep{ga_in_r} against two structurally similar search strategies: simulated annealing (SA) \citep{kirkpatrick} implemented via the base-\texttt{R} \texttt{optim} routine with \texttt{method = "SANN"}, and random-restart greedy hill climbing (HC) \citep{russel2010}. All three methods operate on the same penalised fitness function, constraint set, and the data used in the main illustration. Therefore, the only dimension of variation is the search strategy itself. We evaluate three clubs -- Arsenal (ARS), Manchester City (MNC), and Manchester United (MNU) -- across nine preference-weight values $\lambda_3 \in \{0.1, 0.2, \ldots, 0.9\}$, yielding up to $9 \times 3 = 27$ runs per club.
 
We summarise our analysis in \cref{tab:club_methods}. It is evident that GA achieves full feasibility (9/9 runs) across all three clubs, whereas SA fails to find a feasible solution in 8 of 9 runs for MNC and produces 2 infeasible runs for MNU. HC also appears to be reliable for ARS and MNU but gives 2 infeasible runs for MNC. Further, GA achieves the lowest mean expected transfer cost for ARS ($\text{\euro} 38.9$ million, while HC and SA achieve $\text{\euro} 39.1$ million and $\text{\euro} 47.7$ million respectively) and MNU ($\text{\euro} 55.7$ million, while HC achieves $\text{\euro} 58.1$ million and SA achieves $\text{\euro} 61.3$ million). For MNC, HC gives us the lowest mean cost ($\text{\euro} 113$ million while GA reaches $\text{\euro} 141$ million), but this comparison is unreliable because HC produced only 7 feasible solutions compared to all 9 for GA. In terms of squad quality, GA achieves the highest mean improvement for the two Manchester clubs. For ARS, HC gives us a slightly better squad in comparison to GA. 

\begin{table}[h]
\centering
\caption{Comparison of optimisation methods across clubs}
\label{tab:club_methods}
\renewcommand{\arraystretch}{1.2}

{\scriptsize
\begin{tabular}{lccc|ccc|ccc}
\toprule

& \multicolumn{3}{c|}{{Panel A: Arsenal}}
& \multicolumn{3}{c|}{{Panel B: Manchester City}}
& \multicolumn{3}{c}{{Panel C: Manchester United}} \\

\cmidrule(lr){2-4}
\cmidrule(lr){5-7}
\cmidrule(lr){8-10}

& GA & HC & SA
& GA & HC & SA
& GA & HC & SA \\

\midrule

Feasible runs (out of 9)
& 9 & 9 & 9
& \textbf{9} & 7 & 1
& \textbf{9} & \textbf{9} & 7 \\

Mean Cost
& \textbf{38.9} & 39.1 & 47.7
& 141.0 & \textbf{113.0} & 140.0
& \textbf{55.7} & 58.1 & 61.3 \\

Mean Rating
& 6.66 & \textbf{6.67} & 6.63
& \textbf{7.19} & 7.18 & 7.18
& \textbf{6.50} & 6.49 & 6.47 \\

Mean Run time (mins)
& \textbf{2.10} & 7.97 & 16.40
& \textbf{3.23} & 6.91 & 27.40
& 19.30 & \textbf{6.11} & 18.70 \\

Mean Evaluations (000s)
& \textbf{125} & 730 & 1389
& \textbf{188} & 619 & 2278
& 1109 & \textbf{551} & 1611 \\

\bottomrule
\end{tabular}}
\end{table}

Computationally, GA is substantially more efficient than both alternatives on all three clubs when measured either by wall-clock time or by the number of fitness function evaluations, which is a hardware-independent measure of computational effort. For ARS, GA requires a mean of 125{,}000 fitness evaluations per run compared with 730{,}000 for HC (5.8$\times$ more) and 1{,}389{,}000 for SA (11.1$\times$ more). Similar computational cost holds for MNC (GA: 188{,}000; HC: 619{,}000; SA: 2{,}278{,}000) as well. However, for MNU, GA (1{,}109{,}000) underperforms slightly (HC: 551{,}000; SA: 1{,}611{,}000), but median performance still remains very good, which can be understood from the runtime curves presented in the lower panel of \Cref{fig:comparison_combined}. We also notice that both GA and HC are in agreement with each other in terms of sensitivity analysis with respect to $\lambda_3$. Note that the rating presented in the middle row of \Cref{fig:comparison_combined} is adjusted by subtracting the minimum achieved rating by HC for the sake of scaling. We notice improvement in the rating for both ARS and MNU. However, MNC shows a compressed improvement range for both GA and HC. This can be attributed to the club's good pre-transfer window squad rating (7.168), leaving limited room for improvement. This is also a reason that overall MNC tends to spend more to ensure squad quality is improved, unlike the other two clubs who spend close to zero for lower values of $\lambda_3$. Lastly, in \Cref{fig:pareto_dominance_heatmap}, we present the pairwise dominance for these three methods. We see that SA is dominated by both GA and HC across nearly all comparisons. On the contrary, GA and HC are roughly competitive with each other, with neither consistently dominating the other across the three clubs simultaneously. 
  
\begin{figure}[!htbp]
    \centering
    \includegraphics[width=0.8\textwidth]{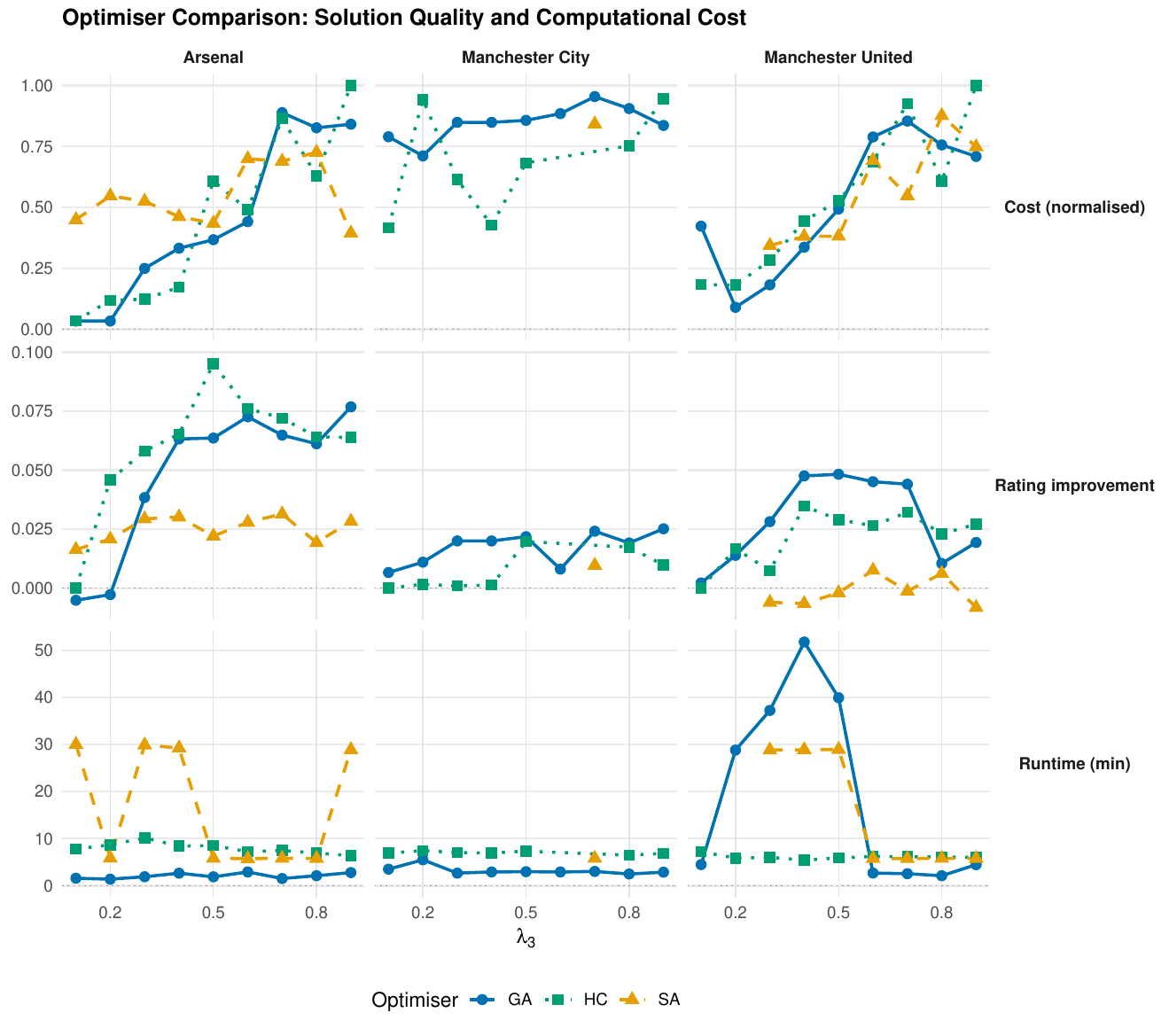}
    \caption{Normalised transfer cost and squad rating improvement by optimiser and quality weight $\lambda_3$, for Arsenal, Manchester City, and Manchester United. Cost is scaled by the maximum observed cost per club; rating improvement is the difference between the recommended new squad mean rating and the minimum squad mean obtained by HC for each time. Runtime is measured in wall-clock minutes and includes any rerun triggered by an infeasible initial solution.}
    \label{fig:comparison_combined}
\end{figure}
 
\begin{figure}[!htbp]
    \centering
    \includegraphics[width=0.8\textwidth]{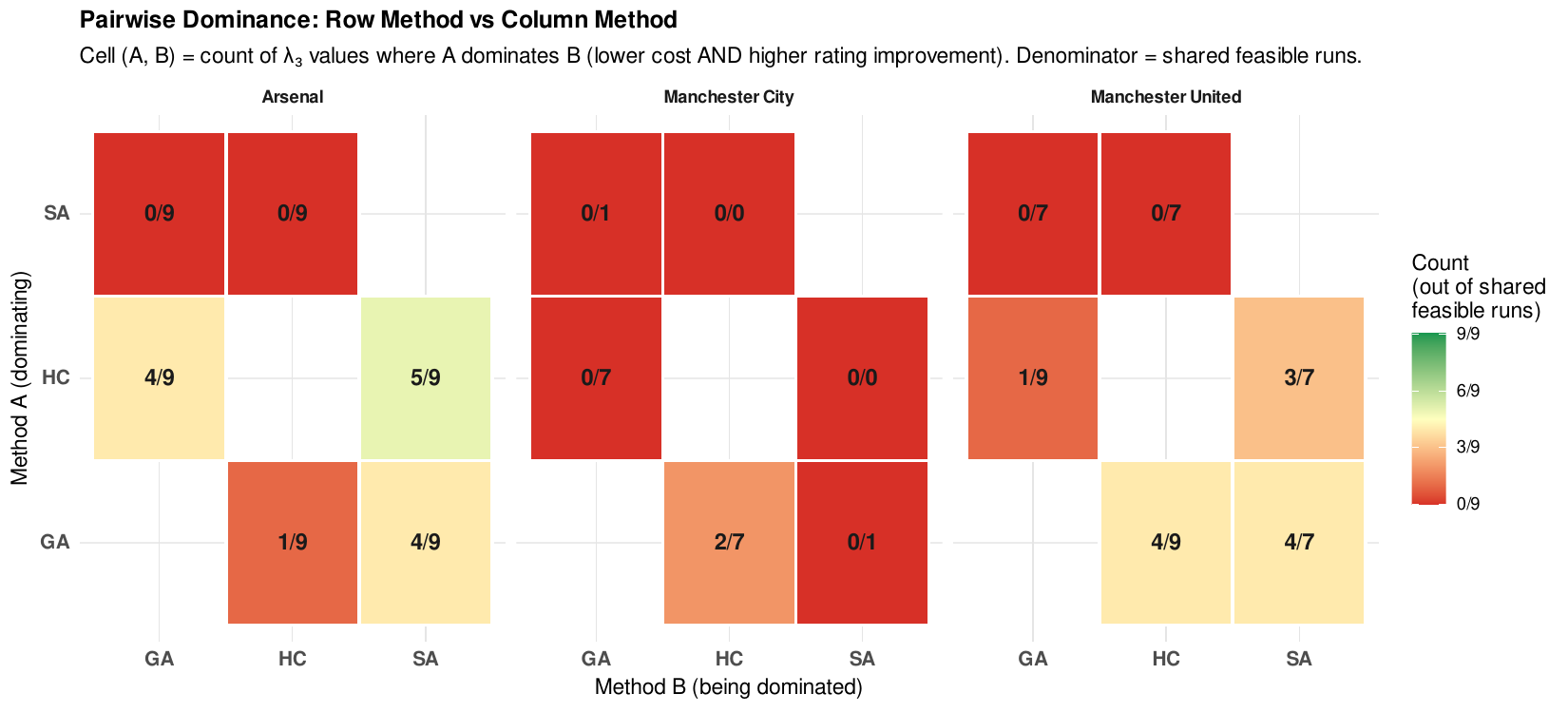}
    \caption{Pairwise dominance matrix for GA, HC, and SA across Arsenal, Manchester City, and Manchester United. Cell $(A, B)$ shows the number of $\lambda_3$ values (out of the shared feasible runs between the two methods) at which method $A$ simultaneously achieves lower normalised transfer cost and higher squad rating improvement than method $B$. A value of $k/n$ indicates that $A$ dominates $B$ on $k$ out of $n$ shared $\lambda_3$ values with a feasible solution. Red cells indicate near-zero dominance; green cells indicate high dominance.}
    \label{fig:pareto_dominance_heatmap}
\end{figure}

Based on these benchmarking results we see that GA achieves full feasibility on all clubs and gives competitive or superior solution quality relative to HC. Besides that, GA requires substantially lower computational time relative to both alternatives. Moreover, the pairwise dominance analysis confirms that GA and HC give us similar Pareto frontiers, indicating the convergence of the GA solution towards the optimal. Therefore, we conclude that GA recommendations can be used for transfer strategies.

\subsection{Additional results on the recommended transfer strategy}

In Section 3.3 of the main paper, we discuss key findings and recommendations from the optimised transfer strategy. This section is dedicated to presenting more detailed results about the same. First, in \Cref{fig:multi:bid}, we present the list of players who were simultaneously identified as attractive transfer targets for multiple English Premier League (EPL) clubs under the proposed framework. The entries correspond to the estimated transfer expenditure (in million euros) associated with acquiring the player for the respective club. The figure illustrates that competition for transfer targets can vary substantially across players. While some players, such as A.~Jahanbakhsh or K.~de Bruyne, are identified as suitable targets by only one club, some players attract interest from multiple clubs simultaneously. In particular, A.~Traore appears to be the most demanded player in the optimisation output, being targeted by Manchester City (MNC), Watford (WAT), Tottenham Hotspur (TOT), and Newcastle United (NEW). Such overlap creates competitive pressure in the transfer market and motivates the auction-theoretic extension discussed in the manuscript. We note that the estimated transfer fees in \Cref{fig:comparison_combined} are model-implied fee benchmarks based on publicly observable variables. In competitive auction settings, the realised prices may exceed these benchmarks substantially, as the equilibrium analysis in Section~3.4 of the main paper and the multi-round extension in \Cref{sec:mult-bid} below show. 

\begin{figure}[h]
    \centering
    \includegraphics[width=0.9\textwidth]{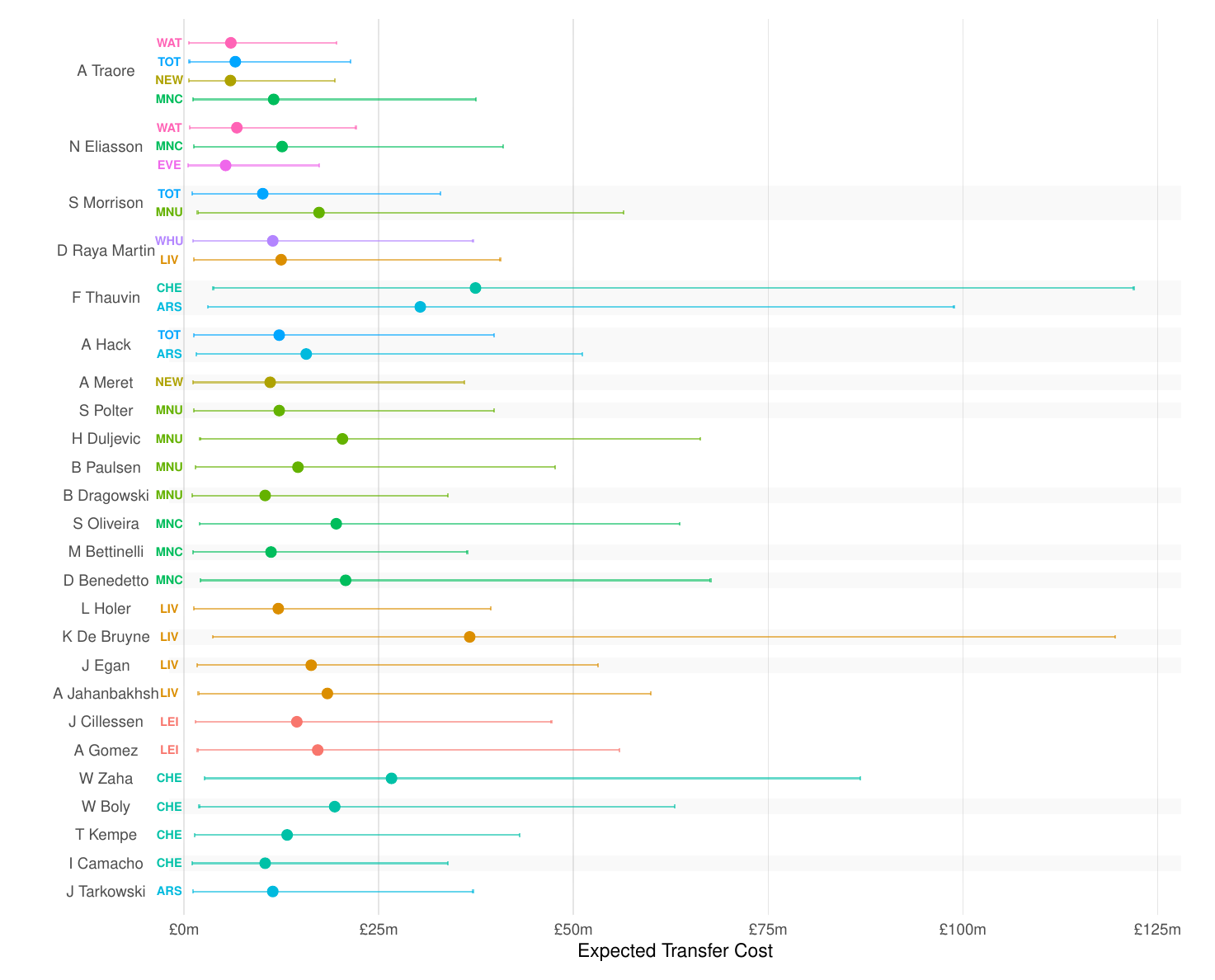}
    \caption{Most expensive players and corresponding bidders according to our optimisation strategy. For each player, the expected price presented by solid points and the 95\% confidence bounds are presented with lines.}
    \label{fig:multi:bid}
\end{figure}

We also notice that the estimated transfer fees vary considerably across clubs for the same player, reflecting the disparity in football transfers. For instance, the projected fee associated with Traore ranges from $\text{\euro} 4.1$ million for TOT to $\text{\euro} 11.5$ million for MNC. Similarly, N.~Eliasson is valued differently by Everton (EVE), MNC, and WAT. These discrepancies arise because the transfer-fee model explicitly incorporates buyer-side effects, league-level financial structures, and team-specific contextual variables. Consequently, the same player may generate different economic valuations depending on the financial strength, squad composition, and tactical requirements of the acquiring club. Moreover, several of the recommendations reveal evidence of strategic differentiation across clubs. Top clubs such as Manchester City, Chelsea, Liverpool, and Arsenal tend to target substantially more expensive players, including F.~Thauvin, K.~de Bruyne, and I.~Pussetto (not presented in the figure due to the large value of the estimated transfer fee). In contrast, clubs operating under tighter budget constraints appear to focus on lower-cost or medium-cost acquisitions that still improve squad balance under the imposed optimisation constraints. This behaviour is consistent with the weighting structure in the optimisation problem, where wealthier clubs place greater emphasis on squad quality while financially constrained clubs place relatively higher emphasis on expenditure and transfer risk.

Another important observation is that the optimisation framework does not merely recommend globally elite players. Instead, many suggested targets are relatively undervalued or niche players whose projected contribution-to-cost ratio fits the structural needs of a particular club. This demonstrates that the framework captures transfer efficiency rather than relying purely on star-player acquisition. In practical terms, the model favours players whose expected marginal contribution to squad quality justifies the associated transfer expenditure under the club-specific financial and positional constraints.

Overall, \Cref{fig:multi:bid} provides additional evidence that football transfers should not be viewed as isolated bilateral negotiations. Instead, transfer outcomes emerge from an interconnected competitive market in which multiple clubs may simultaneously identify the same player as an optimal acquisition target. This interaction between optimisation-based squad construction and strategic market competition provides a key motivation for integrating transfer planning with the auction theoretic framework developed in this paper.

\section{Competitive bidding in multiple rounds}\label{sec:mult-bid}

In football transfer, transfer talks between the seller club and buyer clubs can happen over multiple rounds, as a single round of bidding may not lead to a satisfactory price agreement. So, we extend the single-round auction model to a dynamic game of up to $T$ rounds, in which the buyers' belief about the reserve price evolves after each rejection. Note that, in reality, buyer clubs negotiate with multiple seller clubs for multiple players to achieve their goal of signing a specific type of player, and there is an obvious interaction between the bidding strategies. In this discussion, we refrain from those aspects of the multi-round negotiation paradigm for ease of understanding.

\subsection{Negotiation across rounds}
\label{sec:dynamic}

In general, in the simplest case, we assume that each buyer club submits a bid for their desired player and the seller club evaluates the submitted bids to decide on the transfer of the player. Since, in football, it is customary to have a transfer deadline, we consider that the transfer negotiation proceeds over at most $T$ rounds. The negotiation within round $r$ proceeds as follows. The notations used in this section bear the same meaning as in the main paper.

\begin{enumerate}

    \item Each club $c\in\mathcal{C}^i$ privately observes its valuation $S_c$ and submits a (undisclosed/sealed) bid $b_c = \kappa_c^{(r)}(S_c)$, where $\kappa_c^{(r)}$ is the round-$r$ equilibrium bidding function derived below. Valuations are persistent private information and do not change across rounds.
    
    \item The seller observes the highest bid $\tau^r = \max_{c\in\mathcal{C}^i} b_c$ (set $\tau^0 = 0$ as convention) and the identity of the corresponding club $c^* = \arg\max_c b_c$. The seller then accepts the bid with probability
    \begin{equation}
        a^{(r)} = \widehat{p}_{c^*}(\tau^r)\cdot \frac{H(\tau^r)-H(\tau^{r-1})}{H(\upsilon_{\mathsf{thresh}})-H(\tau^{r-1})},
    \label{eq:accept-prob}
    \end{equation}
    where $\upsilon_{\mathsf{thresh}}$ is the truncation amount (akin to the buyout clause discussed in the paper).
    
    \item If the seller accepts, the player is sold to $c^*$ at price $\tau^r$ and the negotiation ends. If rejected, the seller announces the rejected bid $\tau^r$ publicly, and the bidding moves to round $r+1$. 
    
    \item If no sale occurs after $T$ rounds, the player is not transferred.
    
\end{enumerate}

Clearly, a rejection in round $r-1$ at threshold $\tau^{r-1}$ is informative in nature. It reveals publicly that the seller's reserve satisfies $\rho > \tau^{r-1}$. All buyer clubs will update their beliefs accordingly. Moreover, for the seller, it becomes imperative not to sell at lower prices in the next round. Therefore, the seller's effective reserve distribution in round $r$ is the left-truncated CDF
\begin{equation*}
    \widehat{H}^{(r)}\left(b\right)
    = \Pr \left(\rho \leqslant b \mid \tau^{r-1} < \rho \leqslant \upsilon_{\mathsf{thresh}}\right)
    = \frac{H(b)-H(\tau^{r-1})}{H(\upsilon_{\mathsf{thresh}})-H(\tau^{r-1})},
    \quad b \in \left(\tau^{r-1}, \upsilon_{\mathsf{thresh}}\right],
\end{equation*}
so that $\widehat{H}^{(r)}(\tau^{r-1})=0$ and $\widehat{H}^{(r)}(\upsilon_{\mathsf{thresh}})=1$ for all $r\geqslant 2$. We confirm that the same in round 1 corresponds to $\tau^0=0$, giving $\widehat{H}^{(1)}(b) = H(b)/H(\upsilon_{\mathsf{thresh}})$, which reduces to the buyout-truncated distribution $\widehat{H}(b)$ as described in the manuscript. The hazard rate of this is
\begin{equation}
    \lambda^{(r)}(b) = \frac{h(b)}{H(b)-H(\tau^{r-1})},
\label{eq:hazard-update}
\end{equation}
which is the only structural change to the equilibrium ODE between rounds. The fixed hazard $h(b)/H(b)$ is replaced by $\lambda^{(r)}(b)$ from round $r$ onwards.

Similarly, while club valuations $S_c$ are persistent private information and do not change across rounds, participating in round $r\geqslant 2$ implies truncated support. Any club still participating must have submitted a bid of at least $\tau^{r-1}$ in round $r-1$, which requires a valuation of at least
\begin{equation}
    \underline{s}_c^{(r)} = \psi_c^{(r-1)}\left(\tau^{r-1}\right),
\label{eq:lower-support}
\end{equation}
the minimum valuation that would induce club $c$ to bid $\tau^{r-1}$ in equilibrium. This is a screening condition and not a distributional change. The valuation $S_c$ has not changed, but the club's bidding strategy in round $r$ must be consistent with the fact that $S_c \geqslant \underline{s}_c^{(r)}$. Accordingly, the round-$r$ equilibrium ODE is solved conditional on $S_c \geqslant \underline{s}_c^{(r)}$, which yields the left-truncated distribution
\begin{equation*}
    \widehat{F}^{(r)}_c(s) =
    \frac{F_c(s)-F_c\left(\underline{s}_c^{(r)}\right)}
         {1-F_c\left(\underline{s}_c^{(r)}\right)},
    \qquad
    \widehat{f}^{(r)}_c(s) =
    \frac{f_c(s)}
         {1-F_c\left(\underline{s}_c^{(r)}\right)}.
\end{equation*}
The normalising constant $1-F_c(\underline{s}_c^{(r)})$ cancels in the ratio $\widehat{F}_c^{(r)}/\widehat{f}_c^{(r)}$ entering the ODE, so the only change to the equilibrium condition is replacing $F_c(\psi_c(b))$ with $F_c(\psi_c(b))-F_c(\underline{s}_c^{(r)})$ in the numerator. The lower supports $\underline{s}_c^{(r)}$ are club-specific as the Maskin-Riley ordering implies $\psi_1^{(r-1)}(\tau^{r-1}) < \psi_j^{(r-1)}(\tau^{r-1})$ for $\mu_1>\mu_j$. Thus, stronger clubs require a lower valuation to justify any given bid, implying a lower screening threshold. Since in round~1 there is no prior rejection, for all $c$ we set $\underline{s}_c^{(1)} = \underline{s}$, a common lower bound below which club valuations have negligible probability mass. Now, substituting the updated reserve hazard \eqref{eq:hazard-update} and truncated valuation ratio into the Bayes-Nash first-order condition of \cref{thm:asymm-ode} yields the equilibrium system for the $r^{th}$ round, for $c\in\mathcal{C}^i$, $b\in[b_{\min}^{(r)}, \upsilon_{\mathsf{thresh}}]$:
\begingroup
\begin{small}
\begin{equation}
    \begin{split}
        \psi_c^{(r)\prime}(b) = \frac{1}{\abs{\mathcal{C}^i}-1}\cdot \frac{F_c\left(\psi_c^{(r)}(b)\right)-F_c\left(\underline{s}_c^{(r)}\right)} {f_c\left(\psi_c^{(r)}(b)\right)} &\left[\sum_{j\neq c}\left(\frac{1}{\psi_j^{(r)}(b)-b}-\frac{p_j'(b)}{p_j(b)}\right)\right. \\
        & \left.-(\abs{\mathcal{C}^i}-2)\left(\frac{1}{\psi_c^{(r)}(b)-b} -\frac{p_c'(b)}{p_c(b)}\right) -\frac{h(b)}{H(b)-H(\tau^{r-1})}\right],
    \end{split}
\label{eq:ode-round-r}
\end{equation}
\end{small}
\endgroup

Akin to the arguments to the single-round case of \cref{thm:asymm-ode}, the existence of a monotone Bayes-Nash equilibrium for each round $r$ is satisfied when the truncated reserve hazard $h(b)/\bigl(H(b)-H(\tau^{r-1})\bigr)$ and the truncated valuation ratio $F_c(\psi_c(b))-F_c(s_c^{(r)})$ remain strictly positive on $[b^{(r)}_{\min}, \upsilon_{\mathsf{thresh}}]$. We introduce the bid-gap parameter $\Delta_B^{(r)}$ in \eqref{eq:deltaB} below to ensure this condition holds numerically at the lower boundary. Starting the bid grid at $b^{(r)}_{\min} = \tau^{r-1} + \Delta_B^{(r)}$ ensures the bracket in \eqref{eq:ode-round-r} to be non-negative at the initial grid point. So that the BVP (boundary value problem) collocation procedure begins in a region where a monotone solution exists.

To define the bid-gap parameter, we note that the truncated reserve hazard $h(b)/(H(b)-H(\tau^{r-1}))$ diverges as $b\to\tau^{r-1+}$. Then, applying the first-order Taylor expansion, we get 
\begin{equation*}
    H(\tau^{r-1}+\delta)-H(\tau^{r-1}) = h(\tau^{r-1})\,\delta+O(\delta^2),
\end{equation*}
implying that the hazard is $O(1/\delta)$ near $\tau^{r-1}$. Starting the bid grid at $b_{\min}^{(r)}=\tau^{r-1}+\varepsilon$ for small $\varepsilon$ makes the bracket in \eqref{eq:ode-round-r} potentially negative, signalling that no monotone equilibrium exists at that grid point. With $\mathcal{B}_c^{(r)}$ denoting the bracket in \eqref{eq:ode-round-r} evaluated at $b=\tau^{r-1}+\delta$ with all clubs at their lower supports, we define the bid-gap parameter $\Delta_B^{(r)}$ such that
\begin{equation}
    \Delta_B^{(r)} = \min\left\{\delta>0 :
        \mathcal{B}_c^{(r)}\left(\tau^{r-1}+\delta,\,
        \underline{\boldsymbol{s}}^{(r)}\right)>0
        \;\;\forall\,c\in\mathcal{C}^i
    \right\}.
\label{eq:deltaB}
\end{equation}
This depends on $\tau^{r-1}$ and the reserve parameters $(\mu_\rho,\sigma_\rho)$. At $b_{\min}^{(r)}$, all gap terms $1/(\underline{s}_c^{(r)}-b_{\min}^{(r)})$ enter the bracket symmetrically and the $(\abs{\mathcal{C}^i}-1)$ denominator cancels when checking the sign. Economically, we can argue that $\Delta_B^{(r)}$ represents the minimum bid premium above $\tau^{r-1}$ that a club must offer for the acceptance probability to be non-negligible in round $r$. Then, the system of ODE in \eqref{eq:ode-round-r} is solved on $[b_{\min}^{(r)},\,\upsilon_{\mathsf{thresh}}]$ where $b_{\min}^{(r)} = \tau^{r-1} + \Delta_B^{(r)}$. The lower boundary condition imposes
\begin{equation}
    \psi_c^{(r)}\left(b_{\min}^{(r)}\right) = \underline{s}_c^{(r)},
    \qquad c\in\mathcal{C}^i.
\label{eq:lower-bc}
\end{equation}
No upper boundary conditions are imposed as log-normal distributions have unbounded support, so there is no common finite upper valuation at which all $\psi_c^{(r)}$ will converge to.

\subsection{Numerical solution}
\label{sec:algorithm}

We solve the equilibrium ODE in \eqref{eq:ode-round-r} as a BVP using a collocation approach \citep{ascher1979} rather than an initial value problem since the forward integration from $b_{\min}^{(r)}$ is numerically unstable. The gap term $1/(\psi_c-b)$ is large near $b_{\min}^{(r)}$ where $\psi_c \approx b + \underline{s}_c^{(r)} - b_{\min}^{(r)}$, and small perturbations in the initial condition can lead to degeneracy \citep{fibich2011}. Our proposed BVP approach, therefore, avoids this issue by solving for all grid values simultaneously. For that, we discretise the bid space with a uniform grid of $N=80$ points $b_1 < \cdots < b_N$ over $[b_{\min}^{(r)},\upsilon_{\mathsf{thresh}}]$ with step $h = (b_N -b_1)/(N-1)$. For each club $c$ and each grid point $b_n$, the derivative $\psi_c^{(r)\prime}(b_n)$ is approximated by
\begin{equation}
    \psi_c^{(r)\prime}(b_n) \approx
    \begin{cases}
        \dfrac{\psi_c(b_{n+1})-\psi_c(b_{n-1})}{2h} & n=2,\ldots,N-1,\\[6pt]
        \dfrac{\psi_c(b_N)-\psi_c(b_{N-1})}{h} & n=N.
    \end{cases}
\label{eq:finite-diff}
\end{equation}
Substituting \eqref{eq:finite-diff} into the scaled residual form of \eqref{eq:ode-round-r} we get
\begin{equation*}
    \frac{f_c(\psi_c(b))}{F_c(\psi_c(b))-F_c(\underline{s}_c^{(r)})}\,
    \psi_c^{(r)\prime}(b)
    - \frac{\mathcal{B}_c^{(r)}(b,\boldsymbol\psi)}{C-1} = 0.
\end{equation*}
Together with the lower boundary condition $\psi_c(b_1) = \underline{s}_c^{(r)}$, the above yields a system of $C\times N$ nonlinear algebraic equations in the $C\times N$ unknown grid values $\{\psi_c^{(r)}(b_n)\}_{c,n}$. This system is solved by Broyden's quasi-Newton method \citep{Broyden1965} using the \texttt{R} package \texttt{nleqslv} \citep{nleqslv}. The warm start $\psi_c^{(0)}(b) = b + (\underline{s}_c^{(r)} - b_{\min}^{(r)})$ satisfies the lower boundary condition exactly and ensures $\psi_c>b$ at every grid point, which is required for the gap term $1/(\psi_c-b)$ to remain finite throughout. A solution is accepted only if $\psi_c^{(r)}$ is non-decreasing at all $N$ nodes (tolerance $10^{-5}$), which verifies that a monotone Bayes-Nash equilibrium has been recovered.

However, solving the BVP at every possible $\tau$ encountered during simulation is computationally expensive. Instead, we pre-compute BVP solutions at a grid of $n_\tau$ representative rejection thresholds and interpolate between them during simulation.

We summarise our computation scheme in the list below.

\begingroup
\begin{small}
\setstretch{0.7}
\begin{enumerate}
    \item At first, we set $\tau^0=0$ and $\underline{s}_c^{(1)}=\underline{s}$ for all $c$. Compute $\Delta_B^{(1)}$ via \eqref{eq:deltaB} and solve the round~1 BVP on $[b_{\min}^{(1)},\upsilon_{\mathsf{thresh}}]$.

    \item Afterwards, we simulate 300 independent valuation draws, compute equilibrium bids using $\psi_c^{(1)}$, and record the distribution of $\tau = \max_c\kappa_c^{(1)}(S_c)$. Take the 20\% quantile to 90\% quantile of this distribution as candidate grid points $\{\tau_g\}$.

    \item Finally for each candidate $\tau_g$, in increasing order:
    \begin{enumerate}
        \item Compute the club-specific lower supports $\underline{s}_c^{(r)} = \psi_c^{(r-1)}(\tau_g)$ by evaluating the most recently computed BVP solution at $\tau_g$.
        
        \item Check feasibility: $\underline{s}_c^{(r)} < 0.95\,\upsilon_{\mathsf{thresh}}$ for all $c$, and $\tau_g + \Delta_B^{(r)} < \upsilon_{\mathsf{thresh}}$.
        
        \item Compute $\Delta_B^{(r)}$ using \eqref{eq:deltaB} and solve the round-$r$ BVP on $[\tau_g+\Delta_B^{(r)},\,\upsilon_{\mathsf{thresh}}]$ with lower BC given by \eqref{eq:lower-bc}.
        
        \item Accept the solution only if monotone; otherwise, skip this grid point.
    \end{enumerate}
    Stop when any candidate $\tau_g$ is infeasible.
\end{enumerate}

Once we have our precomputed look-up data, for illustration purposes and uncertainty evaluation, we sample $N_\text{sim}$ number of replicates and for each of $N_\text{sim}$ auctions:
\begin{enumerate}
    \item Draw independent valuations $S_c \sim F_c$ for all
    $c\in\mathcal{C}^i$.

    \item For each round $r = 1,\ldots,T$:
    \begin{enumerate}
        \item Locate the nearest feasible monotone grid point $\tau_g \leqslant \tau^{r-1}$ in the lookup table. Retrieve the corresponding BVP solution $\psi_c^{(r)}$.
        
        \item Compute bids by inverting via linear interpolation: $b_c = \kappa_c^{(r)}(S_c) = [\psi_c^{(r)}]^{-1}(S_c)$ for all $c$.
        
        \item Identify the winner $c^* = \arg\max_c b_c$ and the winning bid $\tau^r = \max_c b_c$.

        \item Compute the acceptance probability from \eqref{eq:accept-prob} and draw $U\sim\mathrm{Uniform}(0,1)$. If $U \leqslant a^{(r)}$, record a sale at price $\tau^r$ and stop. Otherwise record rejection, set $\tau^{r-1} \leftarrow \tau^r$, and continue to round $r+1$.
    \end{enumerate}

    \item If no sale occurs after $T$ rounds, record the path as unsold.
\end{enumerate}
\end{small}
\endgroup

\subsection{Results}

Following our discussions of single-round transfers for two players (M.~Almiron and A.~Traore) in the main manuscript, we extend the analyses with multi-round negotiations for the same players. For that, we keep our initialisation the same as a single round transfer. We set a minimum bid-gap of $\Delta_B=0.7$ to bound the truncated reserve hazard at $O(1)$ throughout. In \cref{tab:hazard-multi}, we report the truncated reserve hazard $h(b_{\min})/\bigl(H(b_{\min})-H(\tau)\bigr)$ at representative rejection thresholds for both players. Without the bid-gap adjustment, the reference hazard $h/H$ (final column) decays toward zero as $\tau$ increases, making the ODE bracket ill-conditioned. The adjusted hazard remains bounded between 1.2 and 2.6 throughout, ensuring numerical stability across all grid points used in the simulation. 

\begin{table}[!h]
\centering
\caption{Reserve hazard at adjusted values and true values for the two players targeted by multiple clubs.}
\label{tab:hazard-multi}
{\scriptsize
\begin{tabular}{lcccc}
\toprule
Player & $\tau$ & $H(b_{\min})-H(\tau)$ &
  $h/[H-H(\tau)]$ & $h/H$ \\
\midrule
\multirow{6}{*}{Almiron}
  & 0.000 & 0.0481 & 2.591 & 2.591 \\
  & 4.207 & 0.0587 & 1.255 & 0.131 \\
  & 5.535 & 0.0441 & 1.261 & 0.085 \\
  & 6.442 & 0.0367 & 1.266 & 0.066 \\
  & 8.233 & 0.0262 & 1.274 & 0.043 \\
  & 9.636 & 0.0205 & 1.280 & 0.032 \\
\midrule
\multirow{6}{*}{Traore}
  & 0.000 & 0.1772 & 1.843 & 1.843 \\
  & 4.207 & 0.0447 & 1.192 & 0.064 \\
  & 5.535 & 0.0283 & 1.212 & 0.038 \\
  & 6.442 & 0.0214 & 1.223 & 0.028 \\
  & 8.233 & 0.0131 & 1.240 & 0.017 \\
  & 9.636 & 0.0093 & 1.250 & 0.012 \\
\bottomrule
\end{tabular}}
\end{table}

After ensuring the stability of the numerical scheme, we perform our analyses of the multi-round negotiation. We present the price distribution across rounds in \cref{tab:price-round-multi} and the top left panels in Figures \ref{fig:almiron:multi} and \ref{fig:traore:multi}. The round~2 price dynamics diverge sharply across the two players. For Almiron, the round~2 mean price ($\text{\euro} 12.8$ million) exceeds round~1 ($\text{\euro} 8.8$ million) by $\text{\euro} 4$ million and the round~2 median ($\text{\euro} 11.6$ million) also exceeds the round~1 mean. For Traore, the round~2 mean ($\text{\euro} 9.2$ million) lies only $\text{\euro} 0.7$ million above round~1 ($\text{\euro} 8.6$m). This contrast is explained by the reserve distributions. Almiron's seller has $\mu_\rho=1.5$, placing substantial probability mass in the range $\text{\euro} 10$ million--$\text{\euro} 16$ million. Bidding paths reaching round~2 for Almiron have experienced a round~1 rejection at $\tau^1\approx\text{\euro} 8$ million--$\text{\euro} 10$ million, so the posterior reserve is concentrated in a high range, and clubs respond with substantially elevated bids. For Traore ($\mu_\rho=0.7$, median reserve $\approx\text{\euro} 2$ million), the posterior mass above round~1 bids is more diffuse, producing only a modest price increase in round~2. The round~3 and round~4 prices for Almiron remain elevated relative to round~1 (means of $\text{\euro} 11.1$ million and $\text{\euro} 10.6$ million), confirming persistent upward price pressure from the reserve-updating mechanism throughout the auction.

\begin{table}[!h]
\centering
\caption{Transfer price (in million euros) conditional on sale occurring in the given round, and rate of sale by rounds conditional on sale not happening before.}
\label{tab:price-round-multi}
{\scriptsize
\begin{tabular}{crrrrrrrr}
\toprule
& \multicolumn{4}{c}{Almiron (3 competing clubs)} &
  \multicolumn{4}{c}{Traore (4 competing clubs)} \\
\cmidrule(lr){2-5}\cmidrule(lr){6-9}
Round & $n$ & Mean & Median & Rate & $n$ & Mean & Median & Rate \\
\midrule
1 & 1,184 &  8.84 &  8.42 & 59.2\% & 1,622 &  8.59 &  8.57 & 81.1\% \\
2 &   126 & 12.83 & 11.59 & 15.4\% &    63 &  9.28 &  9.93 & 16.7\% \\
3 &    34 & 11.13 & 10.19 & 4.9\% &    17 &  8.77 &  8.43 & 5.4\% \\
4 &     8 & 10.62 & 12.14 & 1.2\% &    3 & 10.12 &  9.94 & 1.0\% \\
\midrule
Overall & 1,352 & 9.28 & 8.69 & 67.6\% & 1,705 & 8.62 & 8.57 & 85.3\% \\
\bottomrule
\end{tabular}}
\end{table}

However, it is worth noting that since majority of the sales happen in the first round of negotiations, overall mean transfer price is $\text{\euro} 9.3$ million for Almiron and $\text{\euro} 8.6$ million for Traore, which we illustrate in the top right panels of Figures \ref{fig:almiron:multi} and \ref{fig:traore:multi}. Moreover, for the latter, due to the low buyout clause amount, we can see a spike at $\upsilon_{\mathsf{thresh}} = \text{\euro} 12.3$ million, as we have seen earlier in our analysis with a single round of negotiations (refer to Section 3.4 of the main paper).

We also present the conditional sale rates in \cref{tab:price-round-multi} and the bottom left panels of Figures \ref{fig:almiron:multi} and \ref{fig:traore:multi}. We see that these rates are strikingly similar for both players despite their structural differences. The conditional sale rates are approximately 59--81\% in round~1, 15--17\% in round~2, 5\% in round~3, and 1\% in round~4. This convergence arises because, conditional on reaching round $r$, both auctions face qualitatively similar posterior reserve distributions relative to the surviving bid levels. The sharp decline from round~1 to 2 ($\sim80\%\rightarrow17\%$) reflects the reserve-update screening mechanism: paths reaching round~2 have already experienced a round~1 rejection, revealing that $\rho>\tau^1$ and concentrating the posterior reserve in a region where equilibrium bids remain unlikely to clear it. Across 5 rounds, Traore is sold in 85.3\% of the negotiations versus 67.6\% for Almiron, a gap driven by the difference in reserve distributions ($\mu_\rho=0.7$ versus $\mu_\rho=1.5$). Both figures exceed the corresponding single-round sale rates (71.8\% and 55.8\%), confirming that the multi-round structure raises sale probability by 13--14 \% by offering additional rounds of negotiations to reach a common ground.

\begin{figure}[!htbp]
    \centering
    \includegraphics[width=0.7\linewidth]{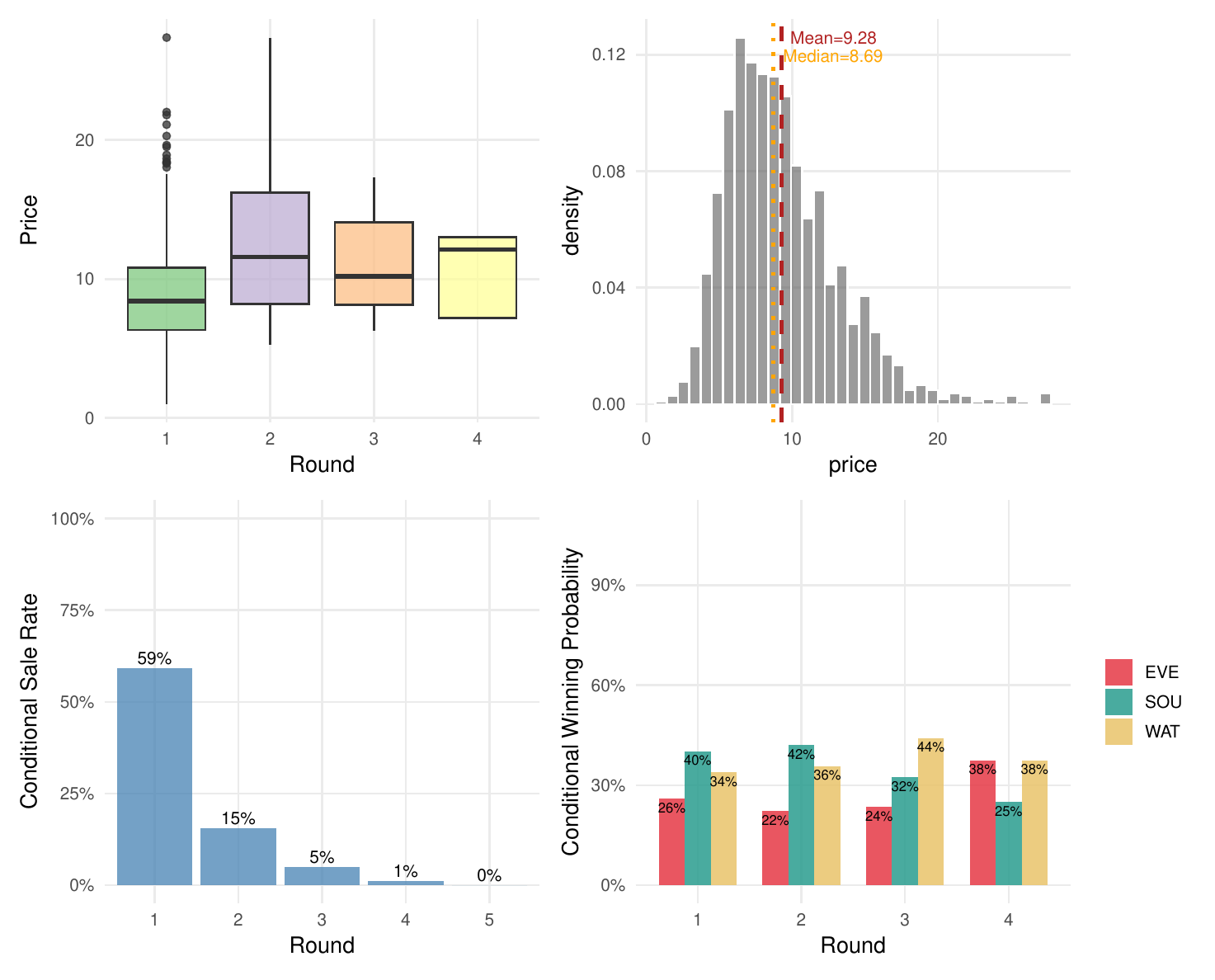}
    \caption{Summary of multi-round negotiation for Almiron: (top-left) round-wise price distribution; (top-right) overall price distribution; (bottom-left) round-wise conditional sale rates; (bottom-right) round-wise win \% for different teams.}
    \label{fig:almiron:multi}
\end{figure}

\begin{figure}[!htbp]
    \centering
    \includegraphics[width=0.7\linewidth]{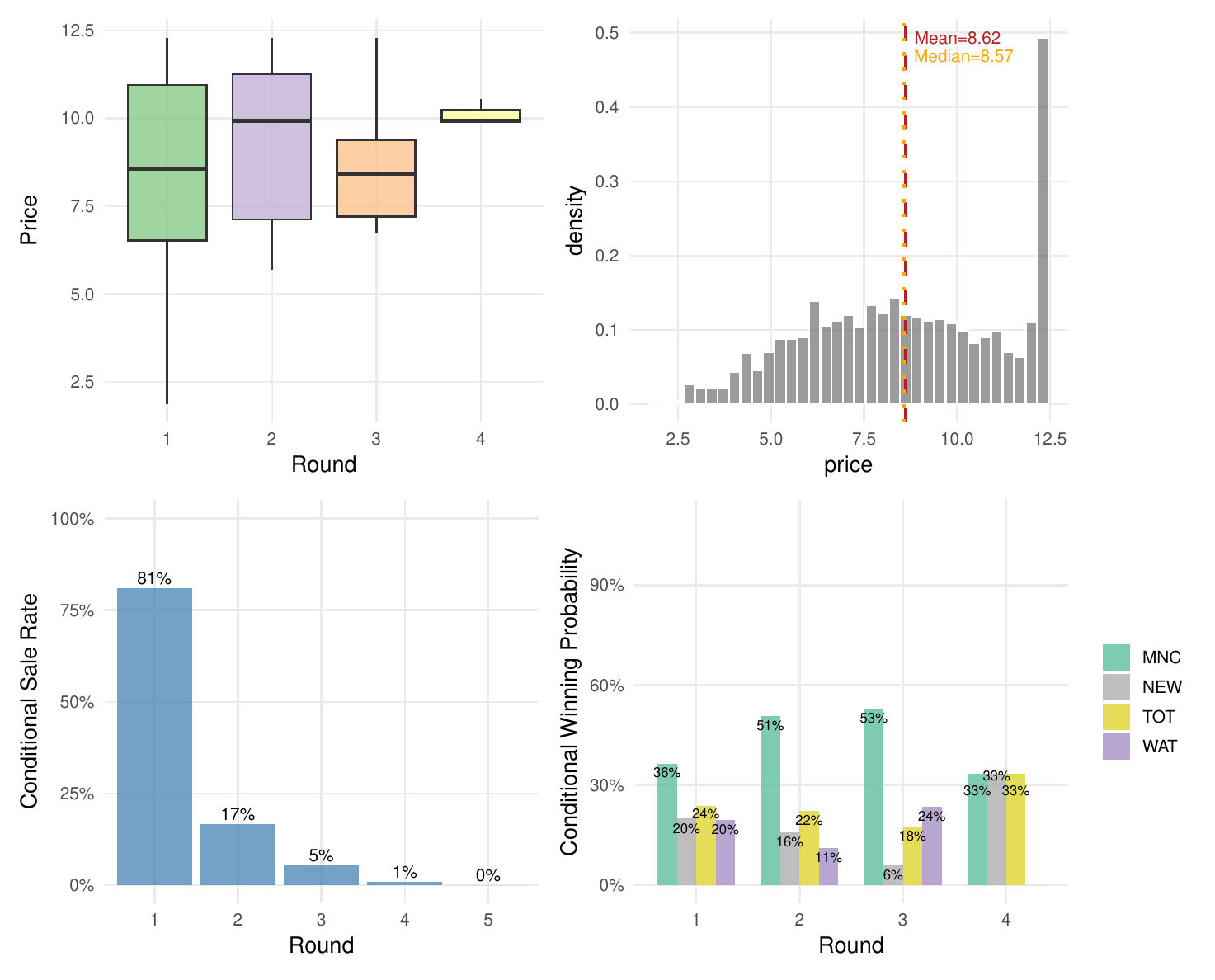}
    \caption{Summary of multi-round negotiation for Traore: (top-left) round-wise price distribution; (top-right) overall price distribution; (bottom-left) round-wise conditional sale rates; (bottom-right) round-wise win \% for different teams.}
    \label{fig:traore:multi}
\end{figure}

Finally, we conclude our illustration by looking at the conditional win probabilities on the bottom right panels of Figures \ref{fig:almiron:multi} and \ref{fig:traore:multi}. We see that the competitive dynamics for Almiron is more or less stable across rounds. Southampton maintains 40.2\%--42.1\% through rounds~1 and 2 before declining to 32.4\% and 25.0\% in rounds~3 and 4. Watford gains progressively: 33.9\% in round~1 rising to 44.1\% in round~3 and 37.5\% in round~4. This partial reversal arises because the three clubs have similar $\mu_c$ values (1.2, 1.5, 1.7) and no single club has an overwhelming bidding advantage over others. Everton's share is remarkably stable across rounds~1--3 (25.9\%, 22.2\%, 23.5\%), suggesting that its lower $\mu_c=1.2$ does not prevent it from competing in hard auctions given sufficient valuation draws. On the contrary, we do not see an equal competition in bidding for Traore. Even though it may seem that Manchester City's conditional win probability for Traore is lower than that of Southampton for Almiron, it is worth mentioning that for the former, four clubs are participating in the competition and MNC dominates the next club by 17\% in the 1st round of bidding which is just 6\% for SOU when they are bidding for Almiron against two competitors. This disparity keeps growing in the next 2 rounds as MNC's conditional win probability reaches 53\% in round~3. The apparent parity in round~4 is not conclusive as only three sales occur in that round, making the computed conditional probabilities unreliable.

\begin{table}[!h]
\centering
\caption{Summary comparison of single-round and multi-round auction outcomes ($N=2{,}000$ simulations per player).}
\label{tab:comparison}
{\scriptsize
\begin{tabular}{lcccc}
\toprule
& \multicolumn{2}{c}{Almiron (3 competing clubs)} &
  \multicolumn{2}{c}{Traore (4 competing clubs)} \\
\cmidrule(lr){2-3}\cmidrule(lr){4-5}
& Single & Multi & Single & Multi \\
\midrule
Sale probability         & 55.8\% & 67.6\% & 71.8\% & 85.3\% \\
Mean price ($\text{\euro} $ million)  &  9.1 &  9.3 &  8.6 &  8.6 \\
Price SD ($\text{\euro} $ million)    &  3.7 &  3.9 &  2.7 &  2.6 \\
Dominant club (share)      & SOU (41.9\%) & SOU (40.1\%) & MNC (39.0\%) & MNC (37.2\%) \\
\bottomrule
\end{tabular}}
\end{table}

Overall, the results suggest that introducing a multi-round negotiation structure can generate meaningful benefits for both sellers and smaller competing clubs. As illustrated in \cref{tab:comparison}, the probability of a successful transfer increases by 11.8 percentage points for Almiron and by 13.5 percentage points for Traore under the multi-round mechanism. This indicates that allowing repeated negotiations improves the likelihood of reaching an agreement between the seller and at least one buyer. For Almiron, the multi-round structure also increases the average transfer fee, implying a higher expected revenue for the seller club, Newcastle United. In contrast, while the average transfer fee for Traore remains broadly unchanged, the higher sale probability still benefits the seller through a greater likelihood of completing the transaction.

Another interesting outcome of the simulation is the competitive balance within the transfer market. Although financially dominant clubs with stronger valuation distributions continue to outperform smaller clubs in both settings, the extent of their dominance tends to decrease under multi-round negotiations. Specifically, the winning share of the dominant bidder falls from 41.9\% to 40.1\% for Southampton and from 39.0\% to 37.2\% for Manchester City. It is rather difficult to make a conclusion with just two case studies, but part of this can be explained by the structure of the affinity function. Seller clubs may tend to be more affine towards the weaker clubs as they are not a direct threat to their aspirations. This, in turn, helps weaker clubs in later rounds of negotiations. Further, it naturally motivates several interesting future directions, including extending the mechanism to simultaneous multi-player/multi-seller negotiations with endogenous entry and drop-out of bidders across rounds, and estimating the affinity function $p_c(\cdot)$ structurally from observed transfer histories between club pairs. A further extension is to endogenise the buyout threshold and contract-related frictions, which would allow the model to jointly explain observed round-by-round bargaining dynamics and the resulting transfer price dispersion across leagues.

\end{document}